\documentclass{article}
\usepackage{arxiv}

\usepackage[utf8]{inputenc}
\usepackage[T1]{fontenc}
\usepackage{booktabs}
\usepackage{amsfonts}
\usepackage{nicefrac}
\usepackage{microtype}
\usepackage{amsmath}
\usepackage{amsthm}
\usepackage{mathtools}
\usepackage{amssymb}
\usepackage{graphicx}
\usepackage{subcaption}
\usepackage{authblk}

\newtheorem{theorem}{Theorem}
\newtheorem{lemma}{Lemma}
\newtheorem{assumption}{Assumption}[section]

\title{Pulsating and rotating spirals in a delayed feedback diffractive nonlinear optical system}

\author[*]{Stanislav Budzinskiy}
\author[*]{Alexander Razgulin}
\affil[*]{Faculty of Computational Mathematics and Cybernetics, Lomonosov Moscow State University}

\begin{document}
\maketitle
\begin{abstract}
We study spiral waves in a mathematical model of a nonlinear optical system with a feedback loop. Starting from a delayed scalar diffusion equation in a thin annulus with oblique derivative boundary conditions, we shrink the annulus and derive the limiting equation on a circle. Based on the explicitly constructed normal form of the Hopf bifurcation for the one-dimensional delayed scalar diffusion equation, we make predictions about the existence and stability of two-dimensional spirals that we verify in direct numerical simulations, observing pulsating and rotating spiral waves.
\end{abstract}


\section{Introduction}
Spirals are common patterns to be observed across various areas of natural sciences, including the Belousov-Zhabotinskiy chemical reaction \cite{WinfreeSpiral1972, WinfreeGeometry2001}, reentrant excitation in cardiac tissues \cite{DavidenkoEtAlStationary1992}, and others. Among the first mathematical results were the analysis and description of logarithmic \cite{CohenEtAlRotating1978} and Archimedian \cite{GreenbergSpiral1980} spirals in the so-called $\lambda-\omega$ systems, which are simple mathematical models of reaction-diffusion processes; even though $\lambda-\omega$ systems are thought to have little correspondence with real physical systems, they appear naturally in asymptotic analysis of general reaction-diffusion equations when Hopf bifurcation happens \cite{CohenEtAlRotating1978, DuffyEtAlSpiral1980}. Asymptotic expressions for planar spirals in the far outer region were obtained in \cite{GreenbergPeriodic1976} on the basis of the dispersion relation for one-dimensional periodic travelling waves. This approach was picked up in \cite{MikhailovKrinskyRotating1983} where spirals in bounded and unbounded circular domains were considered as ensembles of independent one-dimensional periodic oscillators whose phases were connected via a spiral-shape-defining function. For a thorough review of perturbation approaches to spiral waves we address the reader to \cite{TysonKeenerSingular1988}.

In practice spirals tend to lose stability as the tip ceases to trace a circle; this leads to such phenomena as meandering and drifting spirals. It was established numerically in \cite{BarkleyEtAlSpiralwave1990} that meandering spirals bifurcate from rigidly rotating spirals via a Hopf bifurcation (an alternative view was discussed in \cite{GolubitskyEtAlMeandering1997}) and linear stability of the latter was studied in \cite{BarkleyLinear1992}. It was then proposed that the way rigidly rotating spirals lose stability could be explained by invoking the Euclidean symmetry of the plane \cite{BarkleyEuclidean1994}, which lead to an equivariant low-dimensional ODE model describing the appearance of meandering spirals \cite{BarkleyKevrekidisdynamical1994}. This ODE model served as a simple dynamical system describing the spiral dynamics in reaction-diffusion systems but the connection between them needed to be justified; this was achieved in \cite{SandstedeEtAlDynamics1997, SandstedeEtAlCentermanifold1997, SandstedeEtAlBifurcations1999, FiedlerEtAlBifurcation1996} with the use of center manifolds. Center manifolds were further employed for rigidly rotating spirals bifurcating from homogeneous equilibria \cite{ScheelBifurcation1998} and for the period-doubling instability phenomenon \cite{SandstedeScheelPeriodDoubling2007}. Meandering spirals were also studied in the kinematic framework \cite{MikhailovZykovKinematical1991, MikhailovEtAlComplex1994}.

The understanding of spiral spectra is of immense importance for stability analysis and was approached both numerically \cite{BarkleyLinear1992, WheelerBarkleyComputation2006} and analytically \cite{SandstedeScheelAbsolute2000, SandstedeScheelAbsolute2000a}. For instance, spiral spectra were studied on discs and on the whole plane, and it was shown that spectra on discs converge to the union of the absolute spectrum and point eigenvalues as radius tends to infinity \cite{SandstedeScheelAbsolute2000a,WheelerBarkleyComputation2006}.

Another concern is how spirals in reaction-diffusion systems respond to external perturbations, which are inevitable in real physical experiments. External periodic forcing was studied for reaction-diffusion equations themselves \cite{SteinbockEtAlControl1993, ZykovEtAlExternal1994} and for a qualitatively accurate low-dimensional ODE model \cite{MantelBarkleyPeriodic1996} to explain and describe resonant drift dynamics of spirals. A method based on response function of spiral waves was developed in a series of papers \cite{BiktashevaEtAlLocalized1998, BiktashevaBiktashevResponse2001, BiktashevaEtAlComputation2009}; it was then applied to investigate the drift dynamics \cite{BiktashevaDrift2000, BiktashevEtAlOrbital2010, BiktashevaEtAlComputation2010} and the wave-particle duality \cite{BiktashevaBiktashevWaveparticle2003, LanghamBarkleyNonspecular2013, LanghamEtAlAsymptotic2014} of spirals.

While external disturbances may appear on their own, they can also be voluntarily introduced with the aim to control the dynamics of the system: to sustain (or suppress) spiral waves, for example. Periodic external forcing is an example of an \textit{a priori} designed control; however the control may also depend on the state of the system. Control of spiral waves by means of delayed feedback was studied in \cite{GrillEtAlFeedbackControlled1995, ZykovEtAlControlling1997, ZykovEtAlGlobal2004, ZykovEngelFeedbackmediated2004, ZykovEtAlInterference2005}. Such feedback can be used, for instance, to suppress spiral waves in cardiac tissue \cite{PanfilovEtAlElimination2000}.

All the aforementioned papers\textemdash{}whether they treated spirals in unbounded domains or in more physically realistic bounded domains with zero-flux boundary conditions\textemdash{}worked with \textit{systems} of reaction-diffusion equations since it was only through interaction of several substances that Hopf bifurcation could occur. Meanwhile, eigenvalues could be 
forced to be complex by means of different mechanisms. In \cite{DellnitzEtAlSpirals1995}, the authors considered a scalar diffusion equation in a disc with `spiral` (essentially oblique derivative) boundary conditions. Due to the boundary conditions the Laplacian ceases to be self-adjoint and produces complex eigenvalues, making Hopf bifurcation possible. Robin boundary conditions could also be used for spiral and target patterns \cite{GolubitskyEtAlTarget2000}. It is suggested to understand these boundary conditions not as real boundary conditions but as effective boundary conditions or as matching conditions on the edge of the spiral's core.

Optical spirals are among the most vivid examples of self organization of light. They can be experimentally observed even in the simplest optical systems consisting of a thin liquid crystal light valve and a feedback loop, which can be all-optical or may contain digital control elements \cite{AkhmanovEtAlControlling1992}. As the feedback loop is exceptionally flexible in terms of its possible configurations, feedback optical systems are very versatile for studying complex nonlinear phenomena and at the same time are easily controllable. For instance, optical kaleidoscope systems can serve as models for `dry hydrodynamics`: pattern formation of light waves resembles\textemdash{}under certain conditions\textemdash{}complex hydrodynamical flows \cite{VorontsovEtAlSelfOrganization1998}.

Optical spirals can be excited even in one-component systems with feedback, but this requires a combination of local (diffusion and/or diffraction) and nonlocal (rotation of spatial arguments and/or time delay) interaction mechanisms in the feedback loop. Typically, diffusion and argument rotation are used to model spirals \cite{AdachiharaFaidTwodimensional1993, ZheleznykhEtAlRotating1994} yet it is also of interest to include diffraction and delay into the model (they are natural for systems with fast Kerr nonlinearities \cite{AkhmanovEtAlControlling1992}).

For a combination of diffusion and spatial rotation, spiral waves in a disc were suggested to inevitably decay into simpler multi-petal waves \cite{AdachiharaFaidTwodimensional1993}. This was later claimed to be a `wrong conclusion` in \cite{ZheleznykhEtAlRotating1994}, where the authors used the ideas of \cite{MikhailovKrinskyRotating1983} to describe spiral-wave solutions, whose stability, however, was not investigated; the authors only noted that spirals were not the naturally growing modes. This particular problem can be overcome with the help of spiral boundary conditions used in \cite{DellnitzEtAlSpirals1995}. In the absence of nonlocal interactions, though, the now-natural spirals in scalar diffusion equation are unstable with precisely one positive Floquet exponent (which is very close to zero, and so spirals live long in numerical simulations). Thus an additional interaction mechanism could be introduced into the system to push the zero eigenvalue of the Laplacian into the left half-plane.

In this paper, we treat the existence and stability of optical spirals in a thin annulus with spiral boundary conditions. The local interactions include diffusion in the nonlinear Kerr layer and diffraction in the feedback loop; nonlocal interactions are limited to the delay of the control signal in the feedback loop. The mathematical model is thus described by a delayed scalar functional differential diffusion equation. We propose an analytical approach to predict spiral excitation conditions in a thin annulus that is based on studying the Hopf bifurcation in the limiting (as the two-dimensional annulus shrinks) spatially one-dimensional problem on a circle. The latter can be completely understood, as the existence and stability conditions of one-dimensional rotating and standing waves can be expressed in closed form in terms of the physical parameters of the system. We take the theoretical predictions of the one-dimensional model as the basis for the two-dimensional predictions and verify them in direct numerical simulations, observing rigidly rotating and pulsating spiral waves.

\section{Description of the model}
The type of an optical system that we will be looking at consists of a ring cavity formed by a number of mirrors with different reflectivity and a thin layer of Kerr-nonlinear dieletric material. As light enters the resonator, it interacts with the dielectric inducing variations in its refractive index. Then the system can be described by a Debye-type relaxation equation for the phase modulation $u$ that light gets on passing through the dielectric layer \cite{VorontsovFirthPattern1994, AkhmanovEtAlControlling1992, ChesnokovEtAlTimedelayed2002}.

For a system with annular transverse aperture $\Omega$ the equation reads
\begin{equation}\label{eq:2d_model}
    \frac{\partial u}{\partial t} = -u(\rho, \theta, t) + D \Delta u + K \left| \mathcal{B}_{z_0} e^{\mathrm{i} u(t - T)} \right|^2, \quad r < \rho < R, \quad 0 \leq \theta < 2\pi,
\end{equation}
with periodic boundary conditions
\[
    u(\rho, 0, t) = u(\rho, 2\pi, t),\qquad \frac{\partial u}{\partial \theta}(\rho, 0, t) = \frac{\partial u}{\partial \theta}(\rho, 2\pi, t).
\]
Here, $\Delta$ is the Laplacian in polar coordinates. The nonlinear term is proportional to the intensity of the light field at the end of the feedback loop. In our case, a special time-delay device is installed in the ring cavity, and the light propagation operator $\mathcal{B}_{z_0}$ is modeled by a linear Schr\"odinger equation:
\begin{equation}\label{eq:2d_schrodinger}
    \mathcal{B}_{z_0}: A_0(\rho,\theta) \mapsto A(\rho, \theta, z_0; A_0), \qquad \frac{\partial A}{\partial z} + \mathrm{i} \Delta A = 0,\quad A(z=0) = A_0
\end{equation}
When light simply propagates within the annular aperture, Neumann boundary conditions are imposed in the radial direction
\begin{equation}\label{eq:bc_neumann}
    \frac{\partial u}{\partial \rho}(r, \theta, t) = 0,\qquad \frac{\partial u}{\partial \rho}(R, \theta, t) = 0.
\end{equation}
Keeping this as a reference case, we will look at a different type of boundary conditions that is of interest too, the oblique derivative boundary conditions,
\begin{equation}\label{eq:bc_oblique}
    r \frac{\partial u}{\partial \rho}(r, \theta, t) = \tan{\alpha_1} \frac{\partial u}{\partial \theta}(r, \theta, t),\qquad
    R \frac{\partial u}{\partial \rho}(R, \theta, t) = \tan{\alpha_2} \frac{\partial u}{\partial \theta}(R, \theta, t),
\end{equation}
which have two possible meanings: either they are induced on the inner and outer boundaries of the annulus by some special equipment, or they can be used as artificial boundary conditions to describe an annular slice of spiral waves. 

\section{Spiral excitation}
One way to excite spiral waves is through a Hopf bifurcation. To this end, we need to understand spectral properties of the linearized problem, impose Hopf bifurcation conditions on the characteristic values of the equation, reduce the system onto its low-dimensional center manifold, and compute the normal form, whose coefficients are responsible for the qualitative behavior of periodic solutions. Unfortunately, the cubic coefficient cannot be obtained explicitly in closed form, which makes it much more complicated to predict for certain parameters of the model whether stable waves can be excited. To obviate this, we can restrict ourselves to thin annuli and then exploit connections between the thin two-dimensional model and the one-dimensional model on the circle. The annulus being thin, we will assume \[
\alpha_1 = \alpha_2 = \alpha \neq 0.
\]

\subsection{Oblique derivative Laplacian in a thin annulus}
\label{subsection:laplacian}
It is well-known that bifurcation analysis depends prominently on the spectral properties of the linearized problem. For us, it is the Laplace operator in an annulus with oblique derivative boundary conditions. On separating variables, we can observe that the eigenvalues are defined by the zeros of the following cross-product of Bessel functions
\begin{multline*}
	g_{n}({\tan{\alpha}}, \kappa, \zeta) = [J_n'(\zeta)Y_n'(\kappa \zeta) - J_n'(\kappa \zeta)Y_n'(\zeta)] - \mathrm{i}\frac{{\tan{\alpha}}\cdot n}{\zeta}[J_n(\zeta)Y_n'(\kappa \zeta) - J_n'(\kappa \zeta)Y_n(\zeta)] -\\- \mathrm{i}\frac{{\tan{\alpha}}\cdot n}{\kappa \zeta}[J_n'(\zeta)Y_n(\kappa \zeta) - J_n(\kappa \zeta)Y_n'(\zeta)] - \frac{(\tan\alpha)^2 n^2}{\kappa \zeta^2}[J_n(\zeta)Y_n(\kappa \zeta) - J_n(\kappa \zeta)Y_n(\zeta)] = 0,
\end{multline*}
where $n \in \mathbb{Z}$ stands for the angular frequency index, ${\tan{\alpha}}$ measures obliqueness of the boundary condition, and $\kappa = R/r$ measures thickness of the domain. Let $\varepsilon = \kappa - 1$ be small; then
\begin{equation}\label{eq:zero_0}
    \zeta_{n,0}^\varepsilon({\tan{\alpha}}) = n\sqrt{(\tan\alpha)^2 + 1} \left[ 1 - \frac{\varepsilon}{2} - \left( \frac{5(\tan\alpha)^2 - 7}{(\tan\alpha)^2 + 1} + 4 \mathrm{i}\cdot{\tan{\alpha}}\cdot n \right)\frac{\varepsilon^2}{24} + \ldots \right]
\end{equation}
and
\begin{equation}\label{eq:zero_k}
    \zeta_{n,s}^\varepsilon({\tan{\alpha}}) = \frac{s\pi}{\varepsilon} + \frac{ \frac{4n^2 + 3}{8} + \mathrm{i}\cdot{\tan{\alpha}}\cdot n}{\varepsilon + 1}\left(\frac{s\pi}{\varepsilon}\right)^{-1} + \ldots
\end{equation}
with $s \in \mathbb{Z}^+$ the radial frequency index \cite{BudzinskiyZeros2019}. The corresponding eigenvalues are then given by $\lambda_{n,s}^\varepsilon({\tan{\alpha}}) = -(\zeta_{n,s}^\varepsilon({\tan{\alpha}}) / r)^2$. As is commonplace for Laplacians in thin domains, the eigenvalues can be separated in two groups: those that remain finite as the domain shrinks (those on the zeroth frequency in the thin dimension) and those that blow up. The following spectral convergence theorem then follows from \eqref{eq:zero_0} and \eqref{eq:zero_k}.
\begin{theorem}
Let $\lbrace \lambda_q^\varepsilon \rbrace_{q \in \mathbb{N}}$ be $\lbrace \lambda_{n,s}^\varepsilon \rbrace_{n \in \mathbb{Z}, s \in \mathbb{Z}^+}$ enumerated so that $|\lambda_q^\varepsilon| \leq |\lambda_{q+1}^\varepsilon|$. Then as $\varepsilon \longrightarrow 0$, they converge 
\[
\lambda_q^\varepsilon \longrightarrow \lambda_q^0 = -\frac{(\tan\alpha)^2 + 1}{r^2}q^2
\]
to the eigenvalues of scaled second derivative operator on a circle
\[
    \frac{(\tan\alpha)^2 + 1}{r^2} \frac{\partial^2}{\partial \theta^2}
\]
\end{theorem}

\subsection{Bifurcation analysis of the two-dimensional problem}
In our quest for bifurcating spirals we aim to leap promptly to the limit problem; nonetheless, it needs to be checked whether the two-dimensional problem has a center manifold and whether the normal form on it describes the dynamics of the system. To this end we refer to \cite{FariaNormal2000}, where these questions are answered positively given that the linearized operator acting on $u(t)$ generates a "good" semigroup and that the linearized operator acting on $u(t - T)$ is bounded.

The "goodness" of the former is established in
\begin{theorem}[Appendix \ref{appendix:1}]
\label{theorem:2}
The linear operator associated with the oblique derivative boundary value problem in an annulus, acting in $L^2(\Omega)$ and defined on the subspace of $H^2(\Omega)$ of functions satisfying the boundary conditions, generates an immediately compact $C_0$-semigroup.
\end{theorem}

The boundness of the latter appears to be a more delicate issue (for our particular model), and we leave its rigorous treatment for future papers. Instead, we present an informal discussion of the subject in Appendix \ref{appendix:1}.

Though we have not proved our problem to agree with all the assumptions made in \cite{FariaNormal2000}, it seems reasonable to consider pathologies unlikely.

\subsection{Bifurcation analysis of the limit problem}
The one-dimensional model on a circle is given by 
\begin{equation}
    \frac{\partial u}{\partial t} = -u(\theta, t) + \Tilde{D} \frac{\partial^2 u}{\partial \theta^2} + K \left| \mathcal{B}_{\tilde{z}_0} e^{\mathrm{i} u(t - T)} \right|^2
\end{equation}
with periodic boundary conditions
\begin{equation}
    u(0, t) = u(2\pi, t),\qquad \frac{\partial u}{\partial \theta}(0, t) = \frac{\partial u}{\partial \theta}(2\pi, t).
\end{equation}
The Schr\"odinger operator is changed accordingly. Here, the diffusion and diffraction coefficients are scaled as
\[
    \tilde{D} = \frac{(\tan{\alpha})^2 + 1}{r^2}D, \qquad \tilde{z}_0 = \frac{(\tan{\alpha})^2 + 1}{r^2} z_0.
\]
This model was scrutinized in \cite{BudzinskiyRazgulinRotating2017, BudzinskiyRazgulinNormal2017} so we will outline its bifurcation analysis while skipping the technical details.

To study periodic wave solutions bifurcating from a steady state, we localize the problem about the constant solution $u = K$. The nonlinearity parameter $K$ will also serve as the bifurcation parameter so we represent it as $K = \hat{K} + \mu$ with $\mu$ small. We then rewrite the boundary value problem as an abstract dynamical system in the phase space $C\left([-T,0]; H^2_\mathrm{periodic}[0,2\pi] \right)$ of Sobolev-space-valued continuous functions on the delay interval; similar phase spaces are standard in the theory of functional differential equations since the initial data need to be specified on the whole interval.

Next we formulate the Hopf bifurcation conditions that a pair of characteristic values crosses the imaginary axis with nonzero speed at $\mu = 0$. The characteristic equation is understood in the same sense as in the theory of ODEs: we probe the linear part of the equation with exponentials 
\[
h(\theta)e^{\lambda\xi} \in C\left([-T,0]; H^2_\mathrm{periodic}[0,2\pi] \right),
\]
and a characteristic value is then a solution of
\[
-h + \tilde{D} h'' - 2 \hat{K} \mathrm{Im}\left[\mathcal{B}_{\tilde{z}_0}h\right] e^{-\lambda T} = \lambda h, \qquad \lambda \in \mathbb{C}.
\]
\begin{assumption}{(Hopf)}
    There is only one pair of complex conjugate imaginary characteristic values $\lambda_* = \pm \mathrm{i} \nu_* \neq 0$.
\end{assumption}
On expanding $h(\theta)$ into Fourier series 
\[
    h(\theta) = \sum_{n \in \mathbb{Z}}h_n e^{\mathrm{i} n \theta}
\]
we get that either $h_n = 0$ or
\begin{align*}
    &\nu_* = 2\hat{K}\sin(n^2 \tilde{z}_0)\sin(\nu_* T),\\
    &1 + \tilde{D}n^2 + 2\hat{K}\sin(n^2 \tilde{z}_0)\cos(\nu_* T) = 0.
\end{align*}
Since $\nu_* \neq 0$ and $\tilde{D} > 0$, the latter can be satisfied only for a single pair of $n = \pm n_*$. This means that the characteristic values $\lambda_* = \pm \mathrm{i} \nu_*$ are double-degenerate, and this was in fact already known since the equation is $O(2)$-equivariant. Note that the transversality condition is met automatically \cite{BudzinskiyRazgulinNormal2017}.

For the $O(2)$-equivariant Hopf bifurcation the center subspace is four-dimensional and is spanned by 
\[
    \mathrm{span}\lbrace e^{\mathrm{i}(n_*\theta + \nu_* \xi)}, e^{\mathrm{i}(n_*\theta - \nu_* \xi)}, e^{-\mathrm{i}(n_*\theta + \nu_* \xi)}, e^{-\mathrm{i}(n_*\theta - \nu_* \xi)} \rbrace.
\]
The solution is then approximately given by a sum of counter-propagating rotating waves
\[
    u(\theta, \xi) = \hat{K} + \mu + \left( \eta_1 e^{\mathrm{i}(n_*\theta + \nu_* \xi)} + \eta_2 e^{\mathrm{i}(n_*\theta - \nu_* \xi)} + \ldots \right) + \mathrm{c.c}
\]
If we let $\eta_1 = p_1 e^{\mathrm{i}\omega_1}$ and $\eta_2 = p_2 e^{\mathrm{i}\omega_2}$ then the normal form of the $O(2)$-equivariant Hopf bifurcation is
\begin{align*}
    &\dot{p_1} = p_1 \left( A_1 \mu + A_2^{(1)}p_1^2 + A_2^{(2)}p_2^2 \right) + \mathcal{O}(p_1\mu^2 + |(p_1, p_2, \mu)|^4)\\
    &\dot{\omega_1} = \nu_* + \mathcal{O}(|(p_1, p_2, \mu)|)\\
    &\dot{p_2} = p_2 \left( A_1 \mu + A_2^{(1)}p_2^2 + A_2^{(2)}p_1^2 \right) + \mathcal{O}(p_2\mu^2 + |(p_1, p_2, \mu)|^4)\\
    &\dot{\omega_2} = \nu_* + \mathcal{O}(|(p_1, p_2, \mu)|),
\end{align*}
where the coefficients $A_1(\tilde{D}, T, n_*, \nu_*, \tilde{z}_0, \hat{K})$, $A_2^{(1)}(\tilde{D}, T, n_*, \nu_*, \tilde{z}_0, \hat{K})$, and $A_2^{(2)}(\tilde{D}, T, n_*, \nu_*, \tilde{z}_0, \hat{K})$ are known explicitly and in closed form \cite{BudzinskiyRazgulinNormal2017}, and they define the qualitative behavior of the system. In the supercritical case of $\mu > 0$, the most interesting phase portraits correspond to 
\[
A_1 > 0, \qquad A_2^{(1)} < 0, \qquad A_2^{(1)} + A_2^{(2)} < 0.
\]
Then three types of periodic solutions exist: a clockwise rotating wave, a counter-clockwise rotating wave, and a standing wave. Based on the sign of $A_2^{(1)} - A_2^{(2)}$, two cases are possible: either rotating waves are stable, or the standing wave is (see Figure \ref{fig:phase_portrait}).
\begin{figure}
    \centering
    \includegraphics[width=0.7\linewidth]{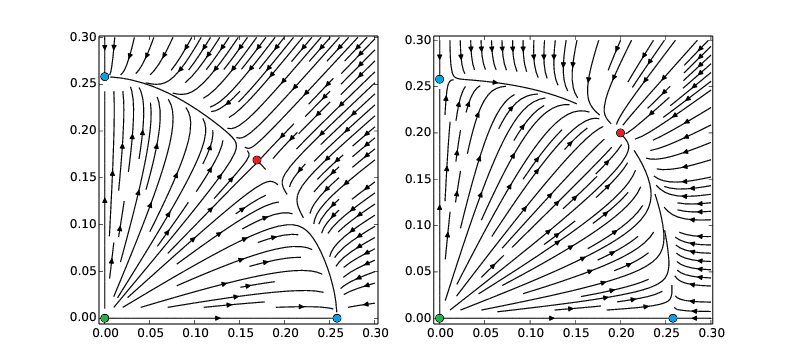}
    \caption{Possible phase portraits of the $O(2)$-equivariant Hopf bifurcation normal form in the $(p_1, p_2)$ coordinates. The dots correspond to: (green) constant solution; (blue) rotating waves; (red) standing wave.}
    \label{fig:phase_portrait}
\end{figure}

\subsection{Application of one-dimensional analysis}
In terms of dynamics, reaction-diffusion equations on thin domains with Neumann boundary conditions in the thin dimension are closely related with their lower-dimensional limits. In the seminal paper \cite{HaleRaugelReactiondiffusion1992} it was proved that attractors are upper semicontinuous as the thin domain shrinks. Persistence of the bifurcation structure was shown for elliptic equations in thin domains \cite{KanPersistence2010}. 

In our previous paper \cite{BudzinskiyEtAlReducing2018} we applied this idea to model two-dimensional bifurcating periodic wave solutions in a delayed feedback nonlinear optical system (despite the theory of functional differential equations on thin domains not being developed yet). Unlike the one-dimensional model on a circle, in the two-dimensional case with Neumann boundary conditions the coefficients of the normal form cannot be computed explicitly even though we know that it has the same structure corresponding to the $O(2)$-equivariant Hopf bifurcation. So we tried to choose the parameters of the two-dimensional system in such a way that its limit one-dimensional problem would have orbitally stable rotating/standing waves, and indeed we managed to excite them in the thin annulus. 

It was noted in \cite{RaugelDynamics1995} that similar behavior is expected from equations with oblique derivative boundary conditions in the thin dimension. However, there are certain nuances. In the limit, the problem acquires an additional symmetry and so goes from $SO(2)$- to $O(2)$-equivariancy. The normal form of the $SO(2)$-equivariant Hopf bifurcation is just the standard Hopf bifurcation and consists of 2 equations; so there are only two coefficients that define the qualitative properties of solutions. Meanwhile \textemdash{} as in the two-dimensional Neumann problem \textemdash{} we cannot evaluate both of them explicitly, and the situation is further worsened since the oblique Laplacian is non-self-adjoint. Thus we decide to use a more complicated normal form of the $O(2)$-equivariant Hopf bifurcation of the limit problem as the source of knowledge about stability of two-dimensional waves.

\section{Numerical experiments}
To excite spiral waves, we will use two sets of parameters that are in Table \ref{tab:params} by scaling them according to the oblique angle and inner radius of the annulus. They guarantee existence and orbital stability of rotating or standing waves, respectively, in the one-dimensional limit problem on a circle. For our experiments we will use $r = 1$, $R = 1.1$, perturbation of the nonlinearity parameter $\mu = 0.1$, and $\alpha \in \lbrace 0, \arctan(\nicefrac{3}{4})\rbrace$; comparing Neumann and oblique cases, we will vividly see how the introduction of obliqueness affects the dynamics of the system.

\begin{table}[h]\centering
\renewcommand{\arraystretch}{1.2}
\begin{tabular}{@{}lllllll@{}}\toprule
    & \multicolumn{6}{c}{\textbf{Parameter values}}\\ \cmidrule{2-7}
    \textbf{Wave type} & $\hat{K}$ & $\tilde{D}$ & $T \, (\tau)$ & $\tilde{z}_0$ & $n_*$ & $\nu_*$ \\ \midrule
    Rotating & 3.286 & 0.068 & 1.817 & 0.0158 & 7 & 1.540 \\
    Standing & 3.662 & 0.214 & 0.592 & 0.06   & 4 & 4.052 \\
    \bottomrule\\
\end{tabular}
\caption{Parameters of the limiting one-dimensional problem on a circle that meet the requirements of the Hopf bifurcation, asymptotic stability of the center manifold, and orbital asymptotic stability of rotating/standing waves.}
\label{tab:params}
\end{table}
To start a simulation, it remains to prescribe the initial data for our initial-boundary value problem, which \textemdash{} owing to the delay $T$ \textemdash{} needs to be given on the whole $[-T,0]$ interval. As the basic waveform we will choose 
\[
    \cos(n (\tan{\alpha}\ln{\rho} + \theta) + \nu t),
\]
which meets the boundary conditions. We shall consider two types of initial data: a pure rotating spiral and a combination of counterpropagating spirals. See Table \ref{tab:ICparams} for more details.
\begin{table}[h]\centering
\renewcommand{\arraystretch}{1.2}
\begin{tabular}{@{}ll@{}}\toprule
    \textbf{Spiral type} & \textbf{Expression} \\ \midrule
    Basic waveform $\mathcal{V}(\alpha, n, \nu)$ & $\cos(n (\tan{\alpha}\ln{\rho} + \theta) + \nu t)$ \\
    Outward rotating & $0.4 \mathcal{V}(\alpha, n_*, \nu_*)$ \\
    Combination of inward and outward & $0.15\mathcal{V}(\alpha, n_*, \nu_*) + 0.25\mathcal{V}(\alpha, n_*, -\nu_*)$ \\
    \bottomrule\\
\end{tabular}
\caption{Different initial conditions to be imposed on $[-T,0]$ for numerical simulations. The basic waveform $\mathcal{V}(\alpha, n, \nu)$ satisfies the oblique derivative boundary conditions with angle $\alpha$. Assuming here that $\alpha \geq 0$, $n_* > 0$, and $\nu_* > 0$.}
\label{tab:ICparams}
\end{table}
To illustrate the behavior of a solution, we will present its snapshots. Each snapshot covers a time interval $[t_1, t_2]$ and shows the evolution of the angular slice in the middle of the annulus $v(\frac{R + r}{2}, \theta, t)$ and the plot of the solution at the end of the time interval $v(\rho, \theta, t_2)$. See Appendix for the numerical method.

\subsection{Rotating parameters}
Let us look at the system with the parameters that give orbitally stable rotating waves in the one-dimensional limit problem (see Table \ref{tab:params}). At first we verify the existence of spiral-wave solutions by commencing the evolution with a pure outward rotating spiral initial condition (see Table \ref{tab:ICparams}). As Figures \ref{fig:rotating_pure_neumann} and \ref{fig:rotating_pure_oblique} demonstrate for $\alpha \in \lbrace 0, \arctan(\nicefrac{3}{4})\rbrace$, the solutions keep their initial form and do not decay. In Figure \ref{fig:rotating_pure_oblique}, the angle-radius plot at the end of the time interval consists of slanting bars, which signify the presence of phase difference between inner and outer circles; the wave is rotating clockwise, hence the spiral is outward. On the contrary, the phase difference is not present in Figure \ref{fig:rotating_pure_oblique} because of the Neumann boundary conditions.
\begin{figure}[t]
    \begin{subfigure}[b]{0.48\textwidth}
        \centering
        \includegraphics[width=\linewidth]{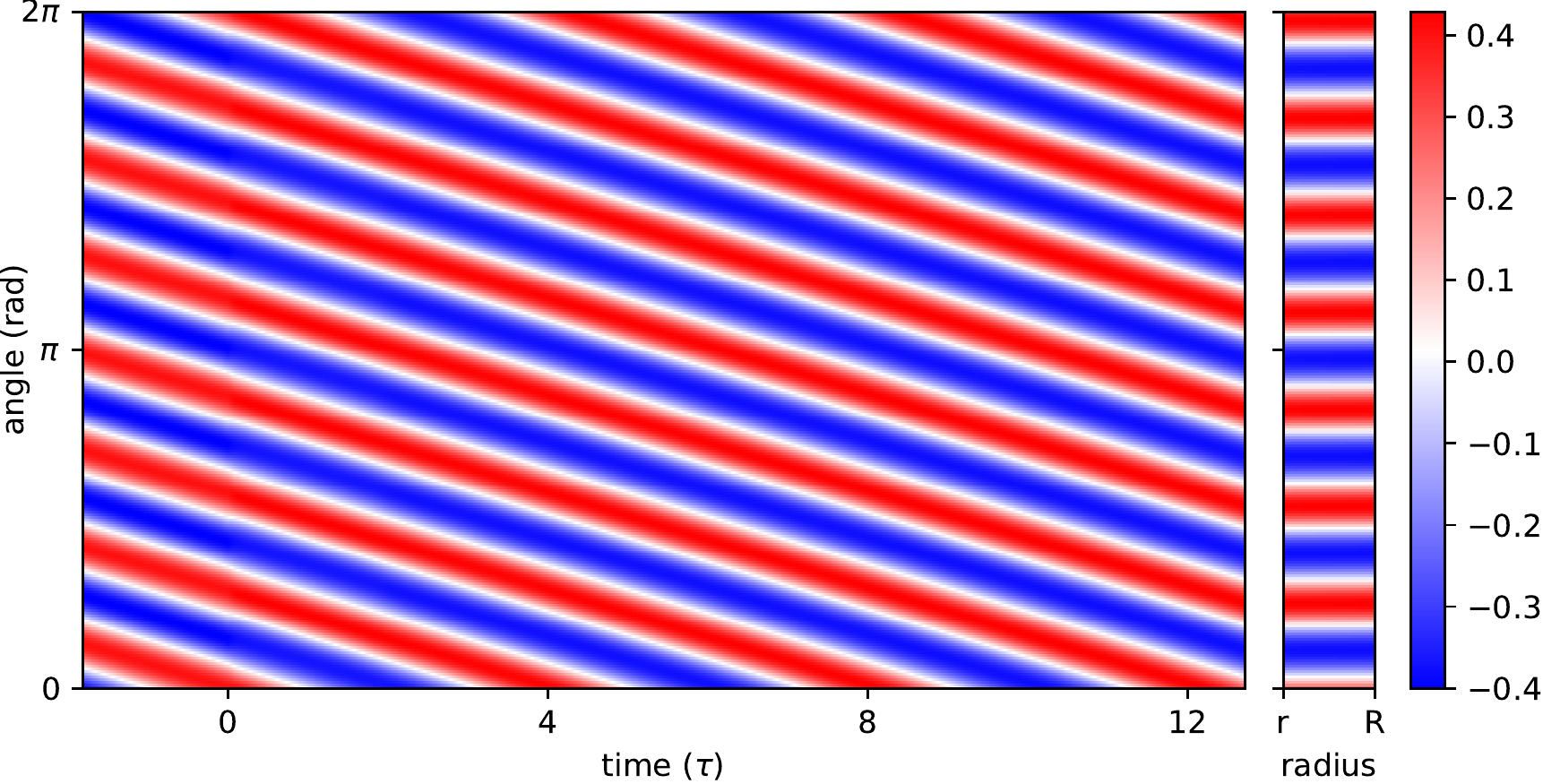}
    \end{subfigure}\hfill%
    \begin{subfigure}[b]{0.48\textwidth}
        \centering
        \includegraphics[width=\linewidth]{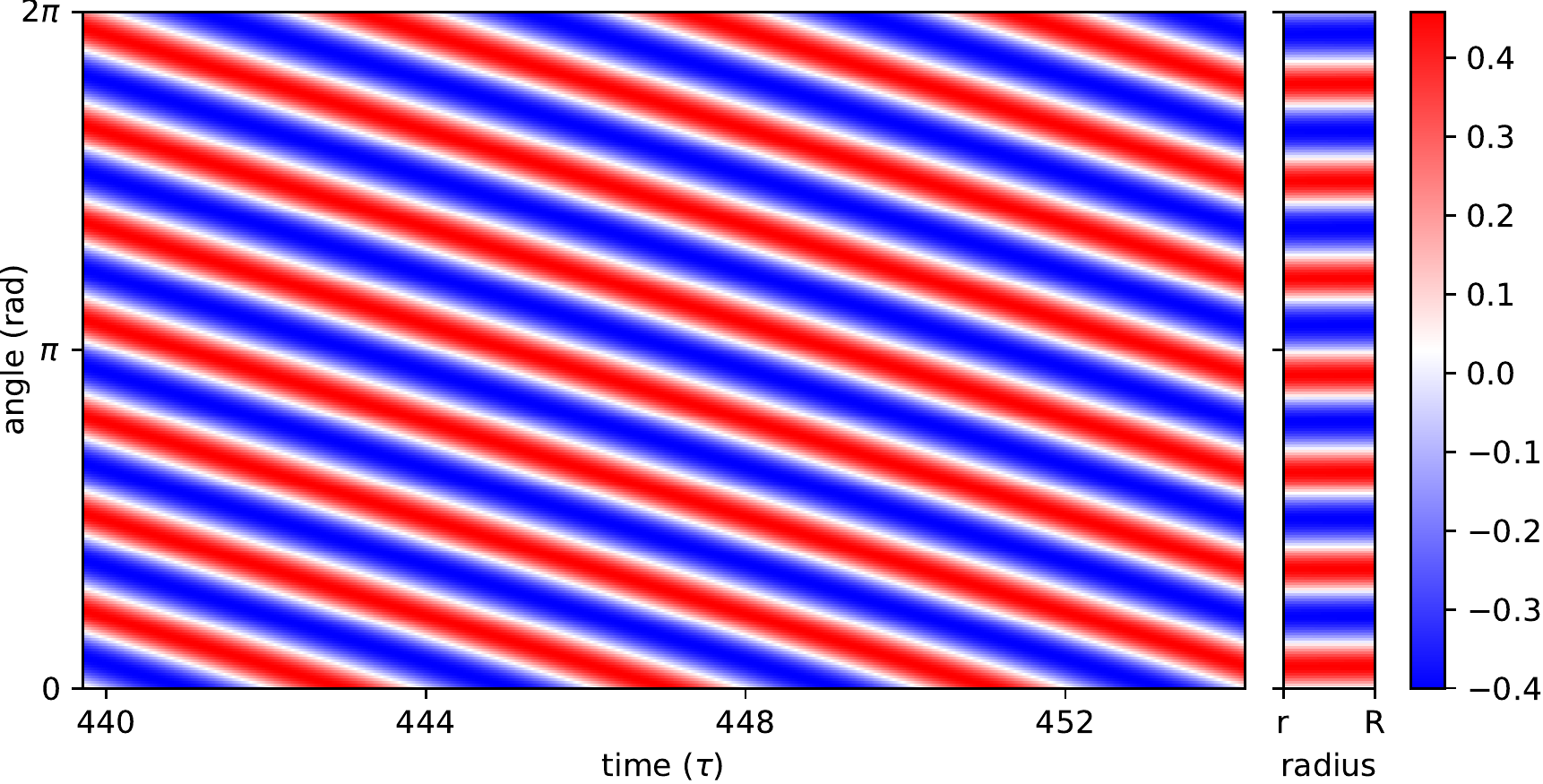}
    \end{subfigure}
    
    \caption{Simulation with rotating parameters (Table \ref{tab:params}), Neumann boundary conditions $\alpha = 0$, and pure rotating spiral initial data (Table \ref{tab:ICparams}). Showing snapshots at $[-T,7T]$ and $[242T, 250T]$.}
    \label{fig:rotating_pure_neumann}
\end{figure}
\begin{figure}[t]
    \begin{subfigure}[b]{0.48\textwidth}
        \centering
        \includegraphics[width=\linewidth]{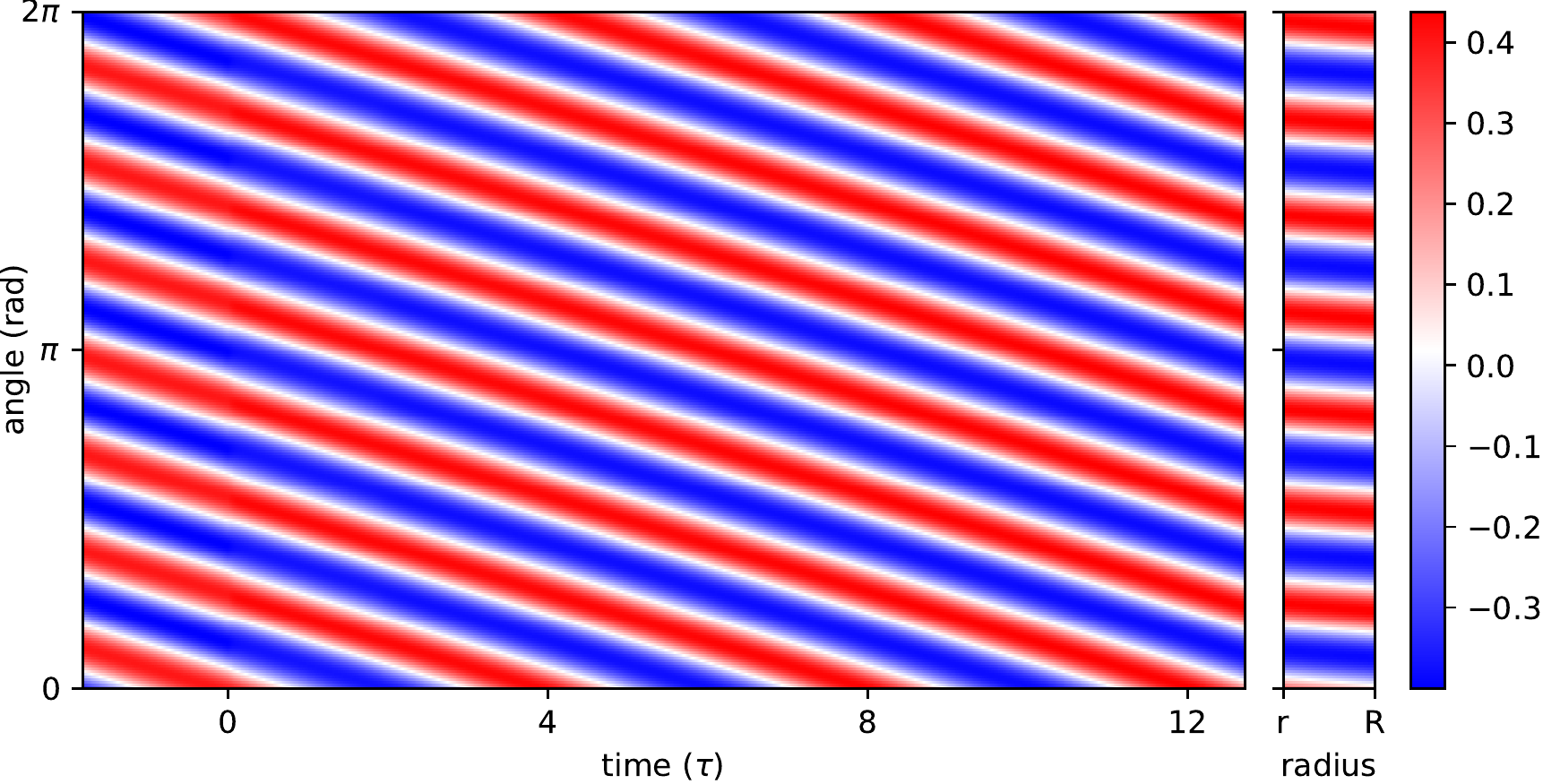}
    \end{subfigure}\hfill%
    \begin{subfigure}[b]{0.48\textwidth}
        \centering
        \includegraphics[width=\linewidth]{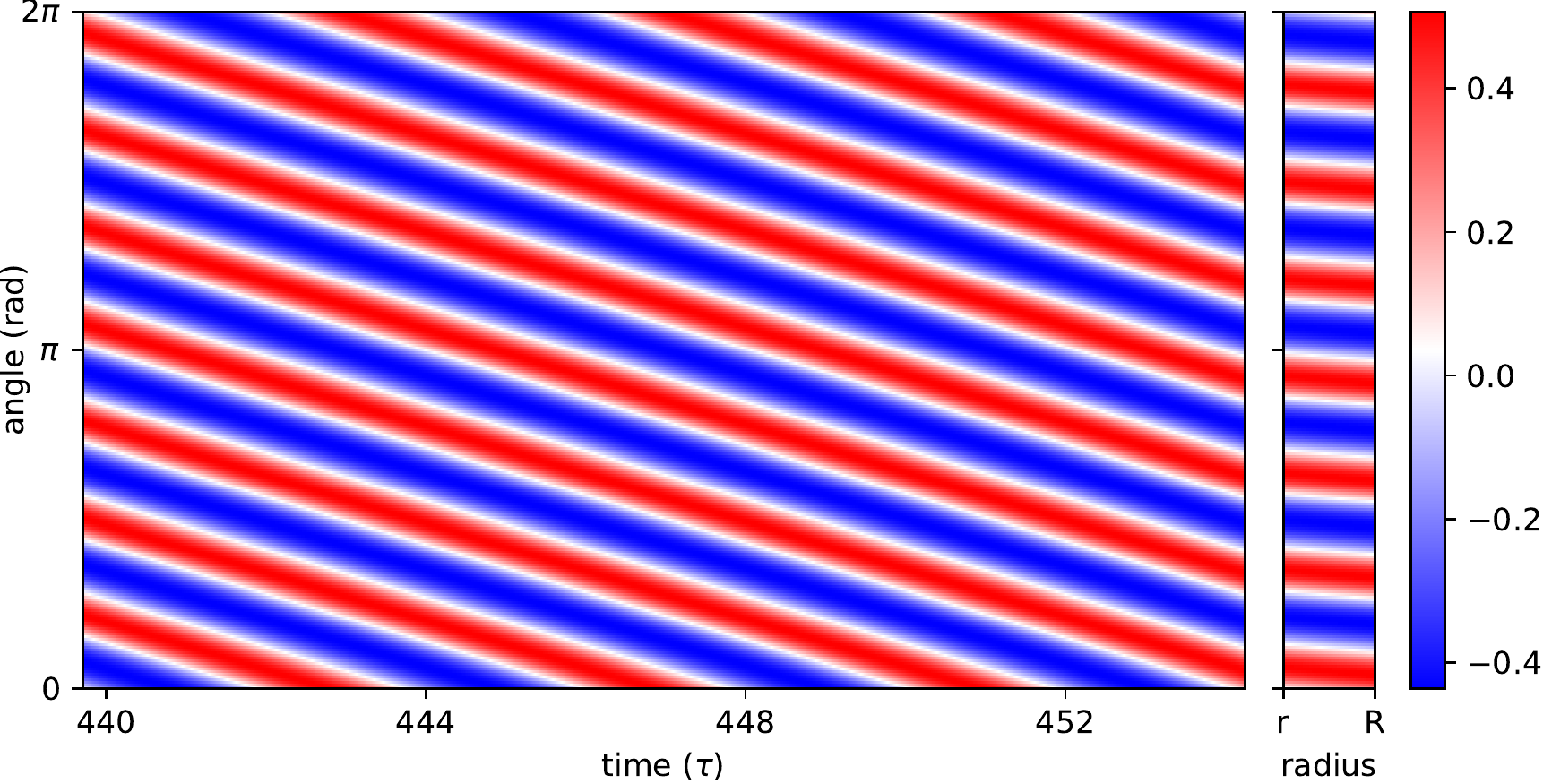}
    \end{subfigure}
    
    \caption{Simulation with rotating parameters (Table \ref{tab:params}), oblique boundary conditions $\alpha = \arctan(\nicefrac{3}{4})$, and pure rotating spiral initial data (Table \ref{tab:ICparams}). Showing snapshots at $[-T,7T]$ and $[242T, 250T]$.}
    \label{fig:rotating_pure_oblique}
\end{figure}

Next, let us turn to stability properties of the spirals. To this end we set the initial data as a sum of inward and outward rotating spirals, where the inward one has a higher amplitude (see Table \ref{tab:ICparams}). The results are shown in Figures \ref{fig:rotating_mix_neumann} and \ref{fig:rotating_mix_oblique}. When $\alpha = 0$, the system is $O(2)$-equivariant and so has no preferable direction of rotation. It exhibits winner-takes-all dynamics: the wave with the higher amplitude snowballs its dominance to dwarf its counter-propagating rival wave and make it disappear. That is, the solution evolves into a clockwise rotating wave in Figure \ref{fig:rotating_mix_neumann}. At the same time, the oblique problem can tell clockwise and counterclockwise directions apart, favoring the counterclockwise, which corresponds to the outward rotating spiral. So in Figure \ref{fig:rotating_mix_oblique}, the solution rotates counterclockwise in the end. This demonstrates that winner-takes-all dynamics is no longer present in the oblique problem and that the outward spiral is attractive.

\begin{figure}[h]
    \begin{subfigure}[b]{0.48\textwidth}
        \centering
        \includegraphics[width=\linewidth]{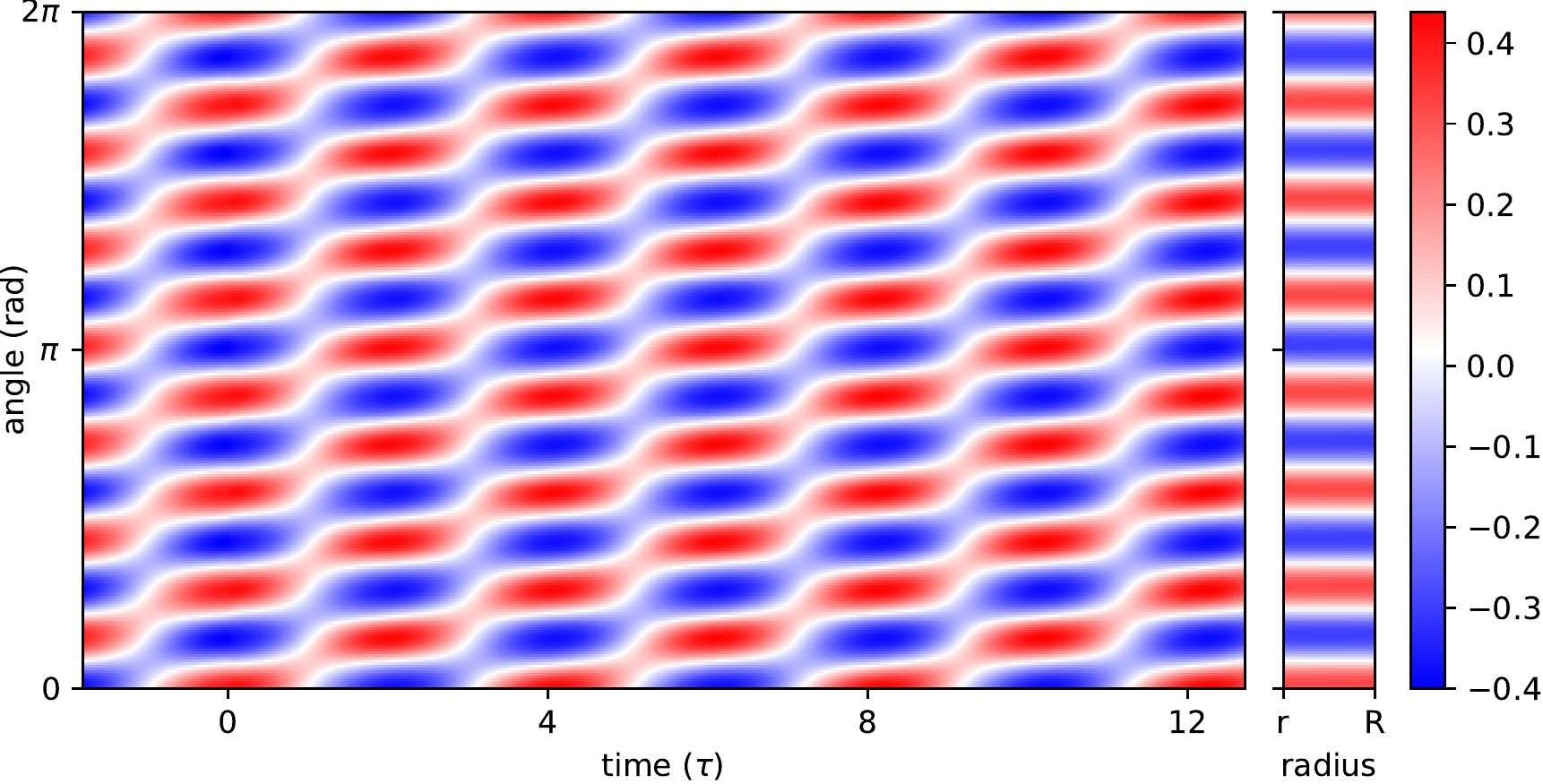}
    \end{subfigure}\hfill%
    \begin{subfigure}[b]{0.48\textwidth}
        \centering
        \includegraphics[width=\linewidth]{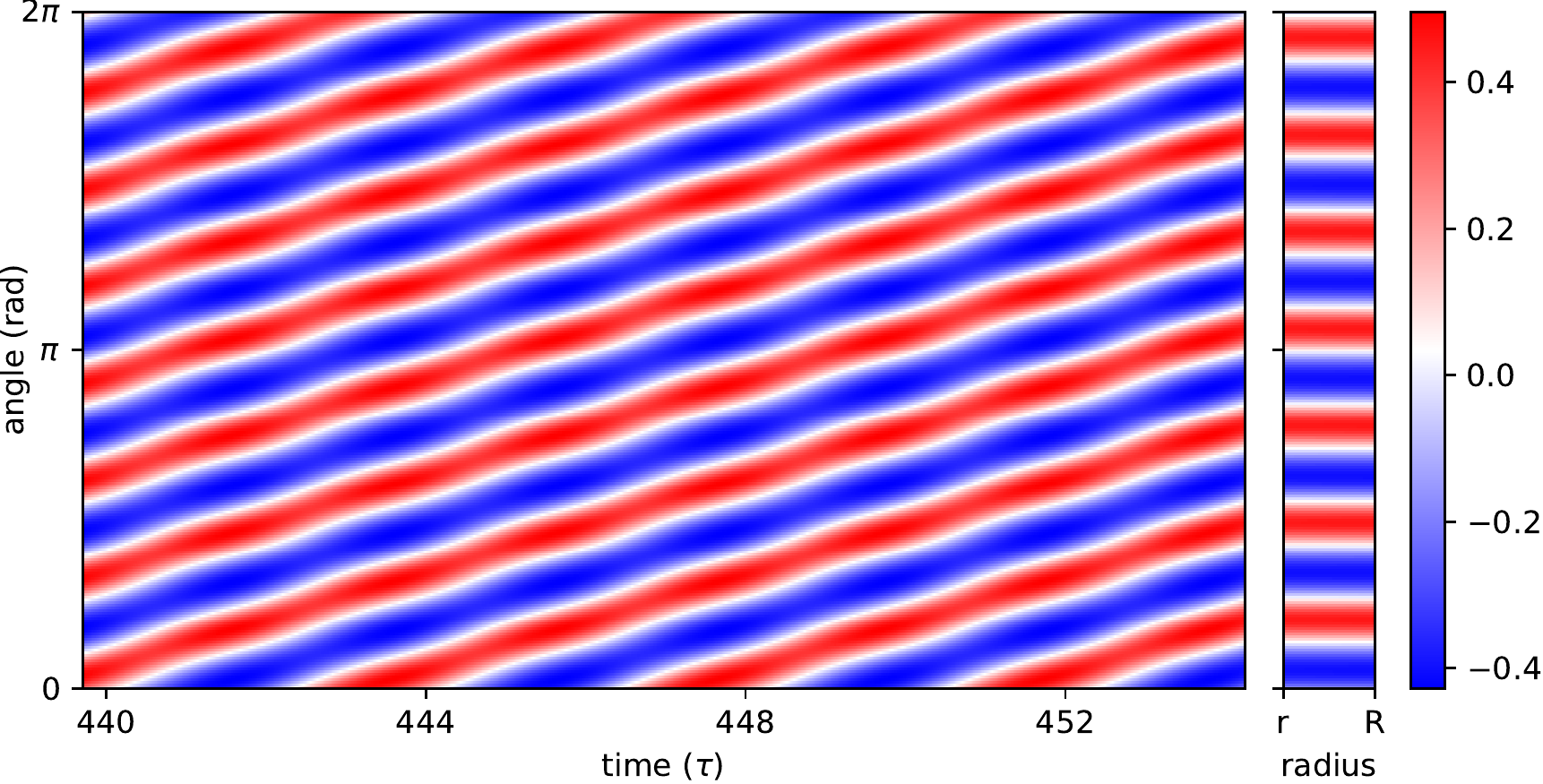}
    \end{subfigure}
    
    \caption{Simulation with rotating parameters (Table \ref{tab:params}), oblique boundary conditions $\alpha = 0$, and mixed initial data (Table \ref{tab:ICparams}). Showing snapshots at $[-T,7T]$ and $[242T, 250T]$.}
    \label{fig:rotating_mix_neumann}
\end{figure}

\begin{figure}[h]
    \begin{subfigure}[b]{0.48\textwidth}
        \centering
        \includegraphics[width=\linewidth]{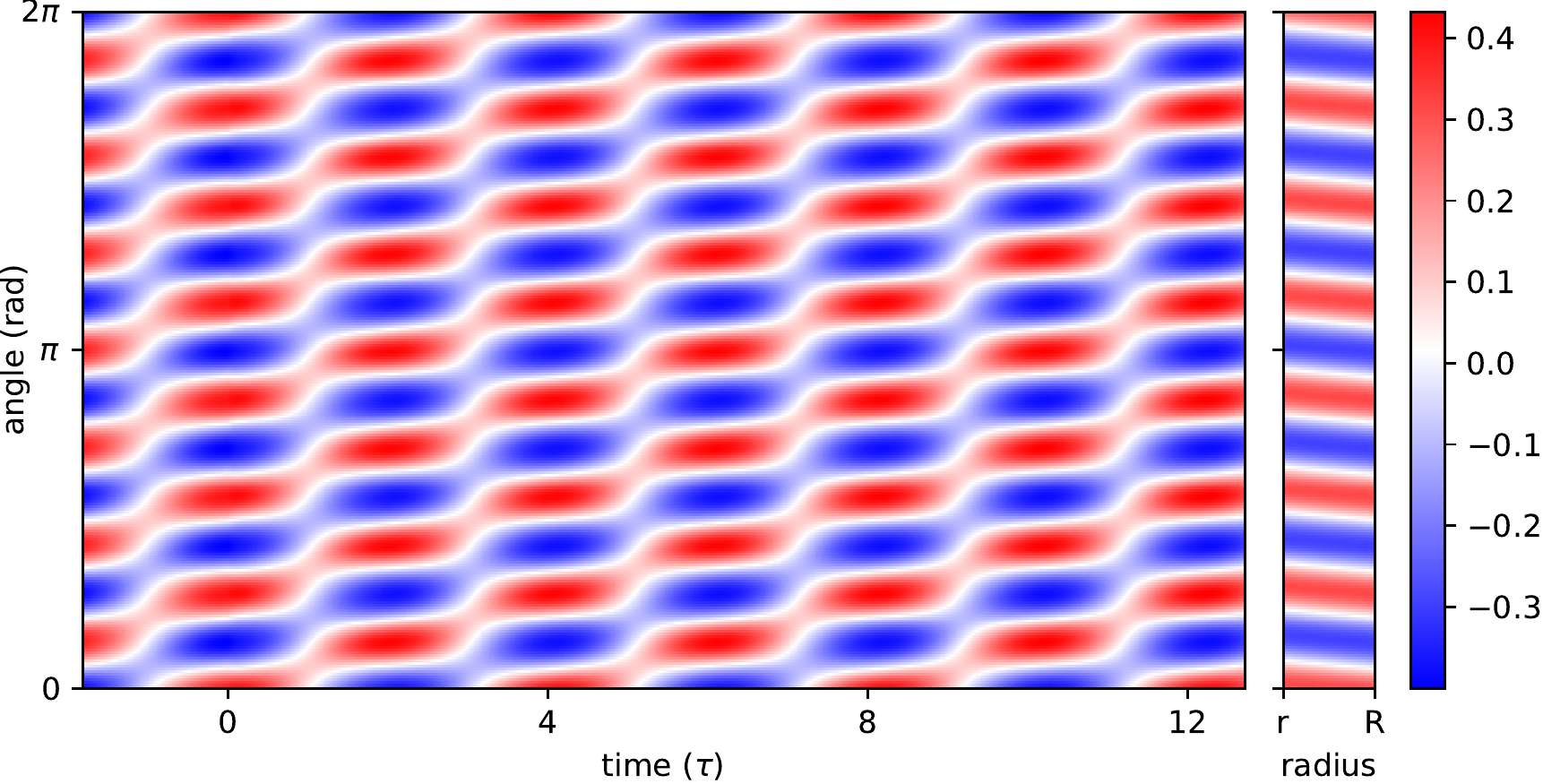}
    \end{subfigure}\hfill%
    \begin{subfigure}[b]{0.48\textwidth}
        \centering
        \includegraphics[width=\linewidth]{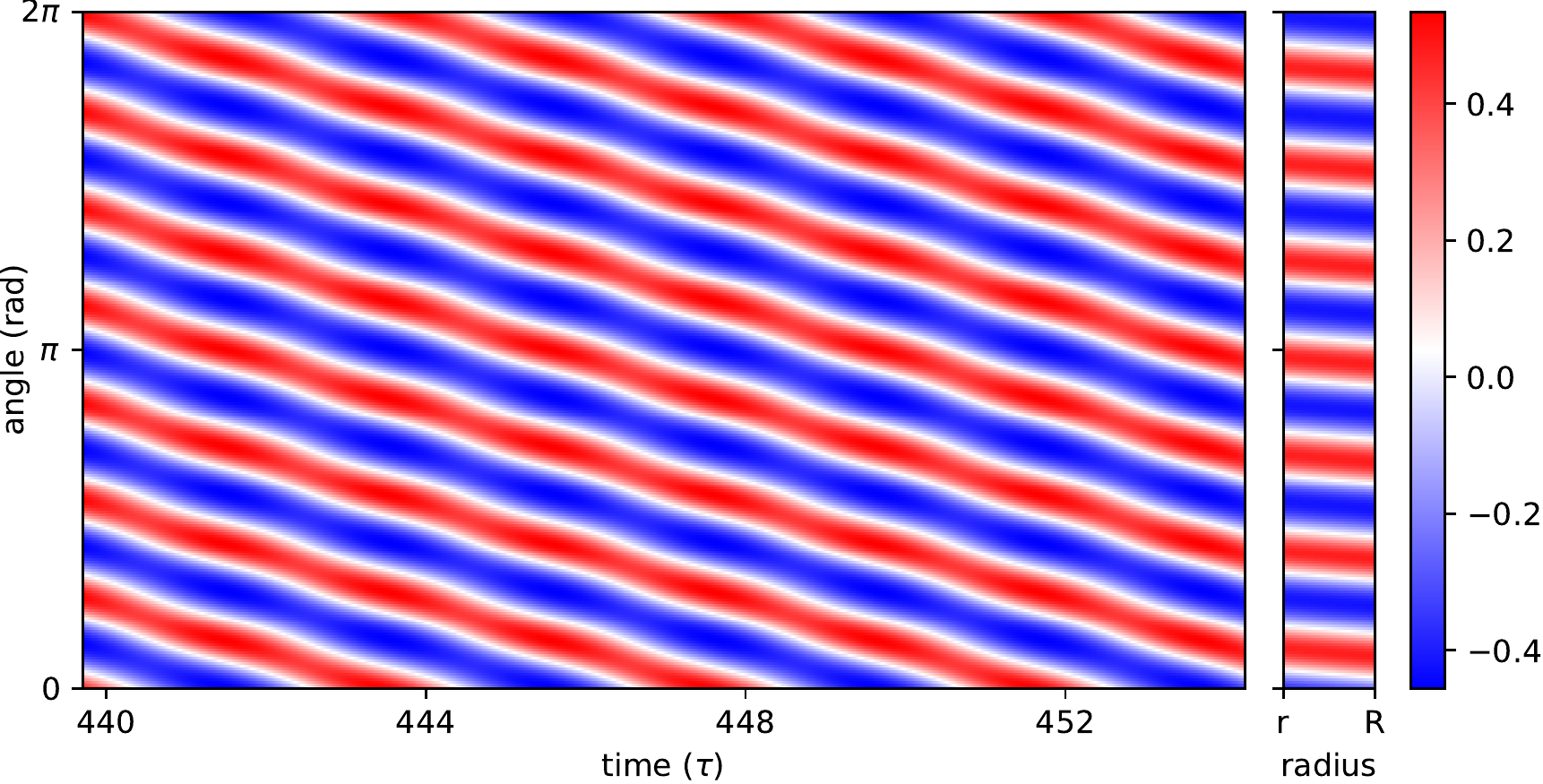}
    \end{subfigure}
    
    \caption{Simulation with rotating parameters (Table \ref{tab:params}), oblique boundary conditions $\alpha = \arctan(\nicefrac{3}{4})$, and mixed initial data (Table \ref{tab:ICparams}). Showing snapshots at $[-T,7T]$ and $[242T, 250T]$.}
    \label{fig:rotating_mix_oblique}
\end{figure}

\subsection{Standing parameters}
Thanks to the reflection symmetry of the one-dimensional limit system, counter-propagating rotating waves can coexist in it in the form of a standing wave, and a suitable choice of parameters (see Table \ref{tab:params}) can make this state orbitally stable. The system then demonstrates cooperative dynamics: in an uneven combination of rotating waves, they work towards balancing each other out. 

The same holds true for the two-dimensional system with Neumann boundary conditions since it retains the $O(2)$ symmetry group. In Figure \ref{fig:standing_mix_neumann}, we start from a combination of waves in the initial data to see them even out as a standing wave. Meanwhile, oblique boundary conditions break the reflection symmetry and so standing wave solutions cease to exist. Cooperative mindset in conjunction with the preference of outward rotation act as checks and balances: they result in pulsating outward rotating spiral (see Figure \ref{fig:standing_mix_oblique}).

\begin{figure}
    \begin{subfigure}[b]{0.48\textwidth}
        \centering
        \includegraphics[width=\linewidth]{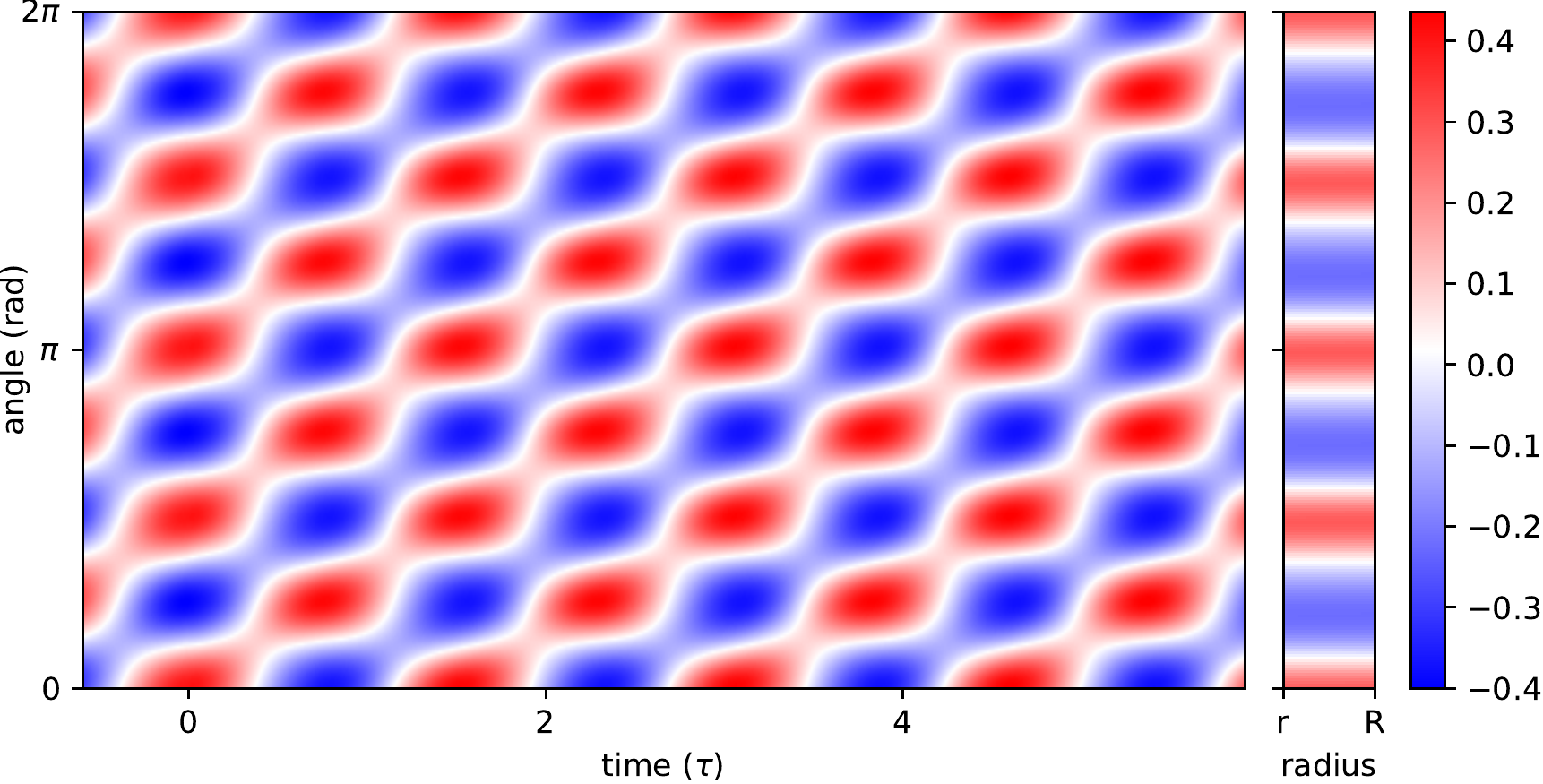}
    \end{subfigure}\hfill%
    \begin{subfigure}[b]{0.48\textwidth}
        \centering
        \includegraphics[width=\linewidth]{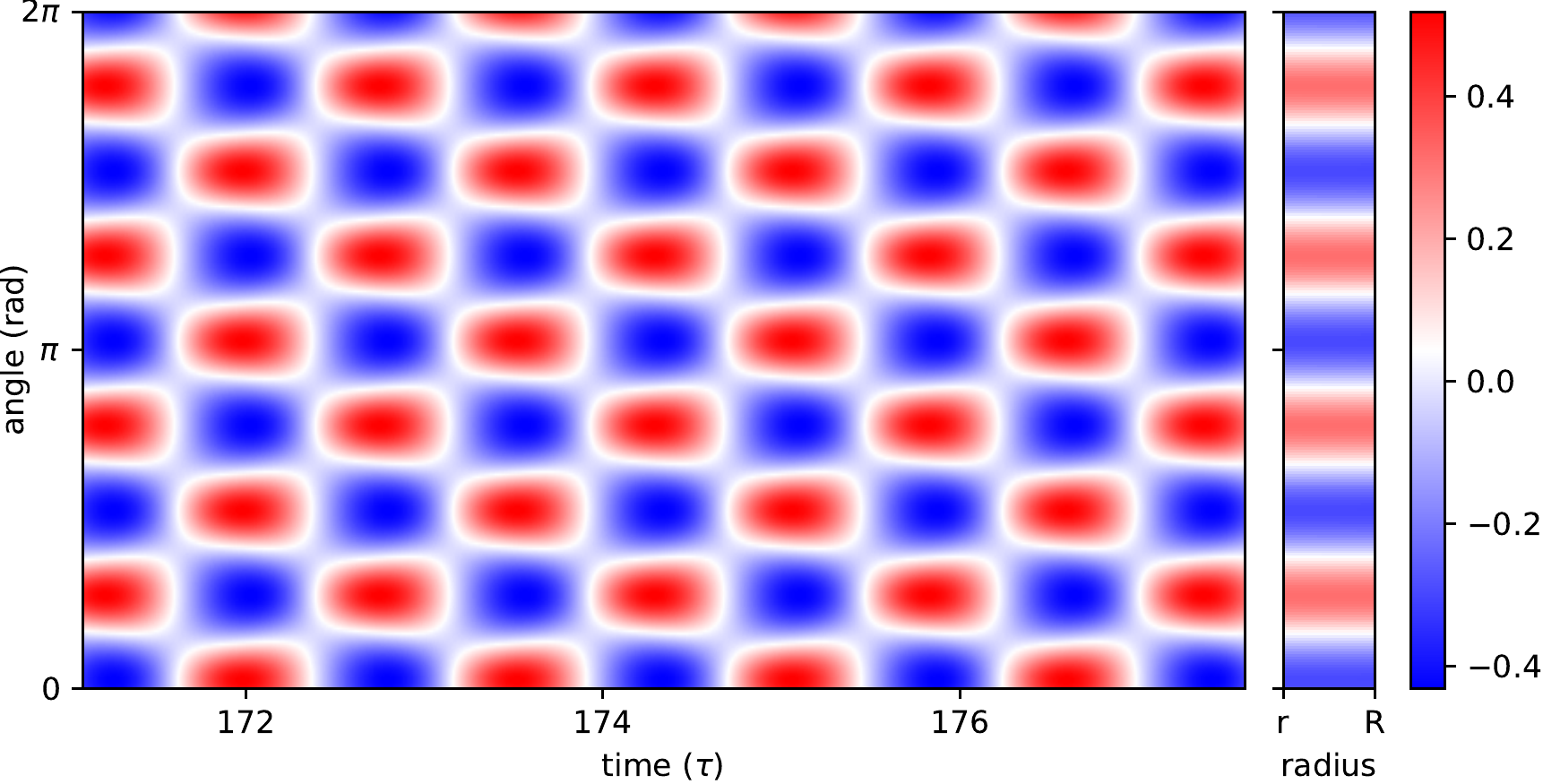}
    \end{subfigure}
    
    \caption{Simulation with standing parameters (Table \ref{tab:params}), Neumann boundary conditions $\alpha = 0$, and mixed initial data (Table \ref{tab:ICparams}). Showing snapshots at $[-T,10T]$ and $[289T, 300T]$.}
    \label{fig:standing_mix_neumann}
\end{figure}

\begin{figure}
    \begin{subfigure}[b]{0.48\textwidth}
        \centering
        \includegraphics[width=\linewidth]{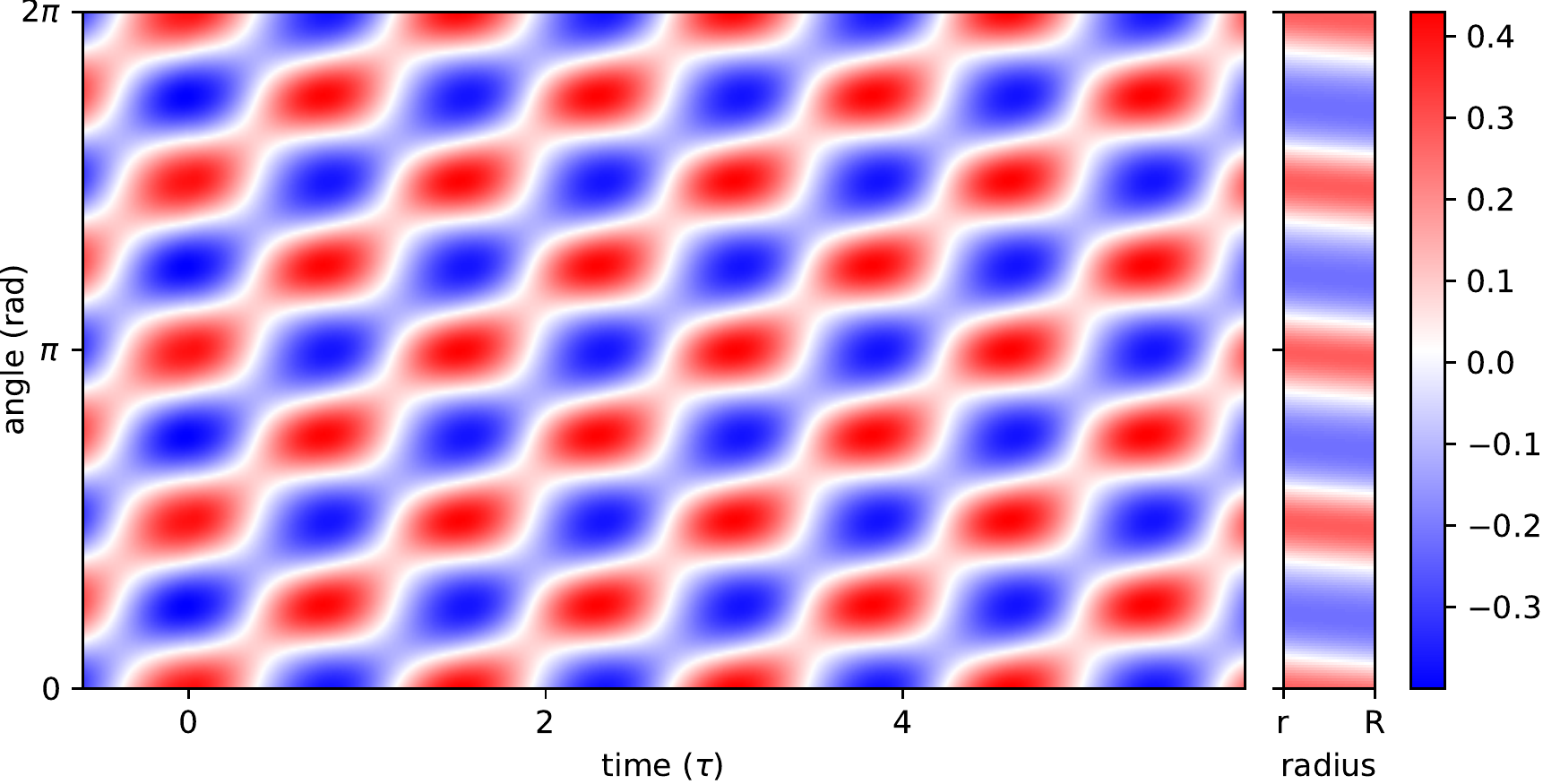}
    \end{subfigure}\hfill%
    \begin{subfigure}[b]{0.48\textwidth}
        \centering
        \includegraphics[width=\linewidth]{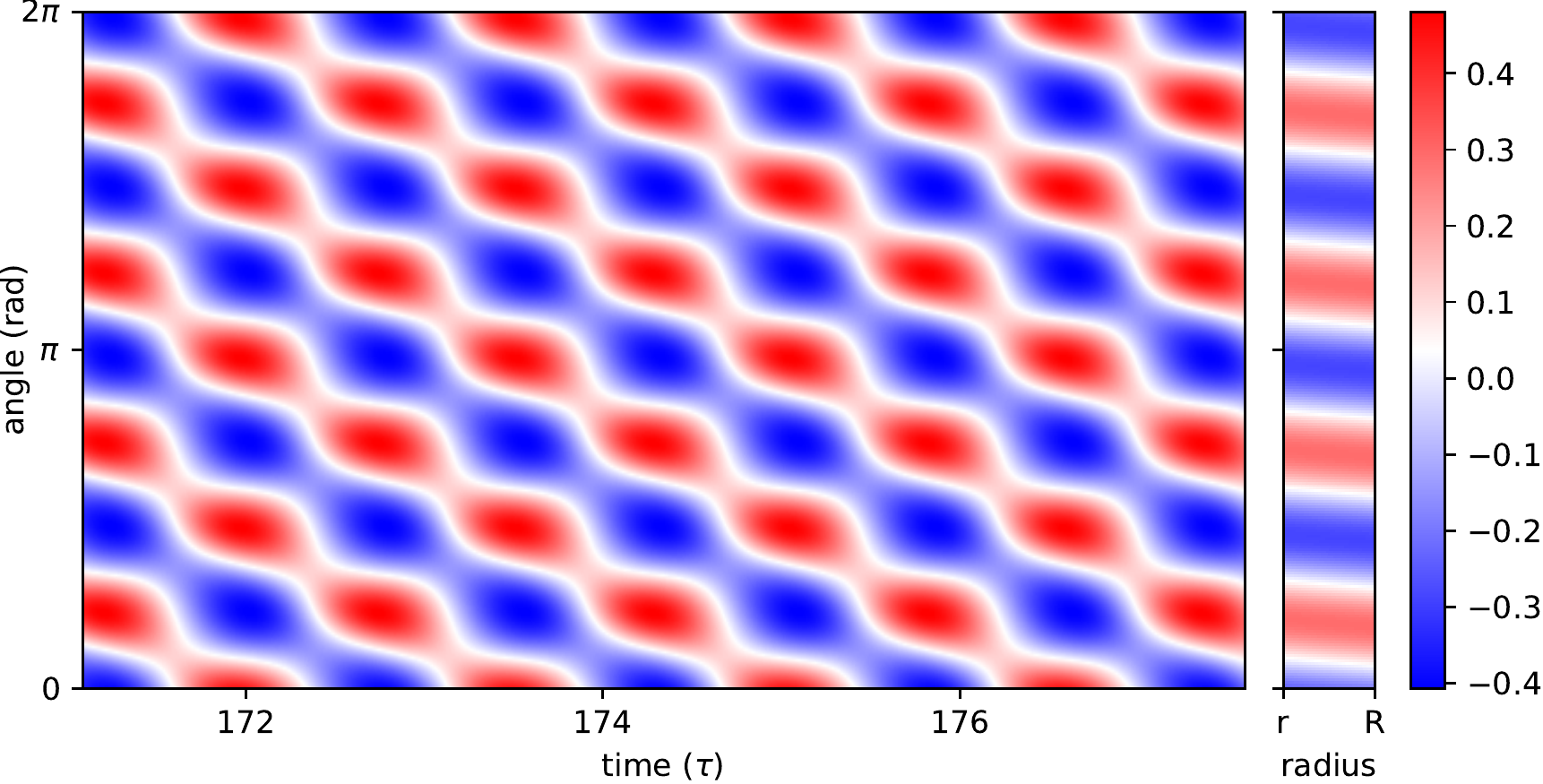}
    \end{subfigure}
    
    \caption{Simulation with standing parameters (Table \ref{tab:params}), oblique boundary conditions $\alpha = \arctan(\nicefrac{3}{4})$, and mixed initial data (Table \ref{tab:ICparams}). Showing snapshots at $[-T,10T]$ and $[289T, 300T]$.}
    \label{fig:standing_mix_oblique}
\end{figure}

By a pulsating spiral we mean an uneven sum of inward and outward rotating spirals. For instance, the initial data we used is also a pulsating wave though inward rotating. The behavior of such waves is rather peculiar: they display in-phase periodic oscillations of amplitude and rotation speed; when the amplitude is high, they rotate very slowly; as the amplitude reaches its minimum, they dash in the same direction and then `come to a halt` to regain their amplitude. 

Pulsating waves were observed in the two-dimensional system with Neumann boundary conditions too but only as an intermediate stage (possibly long lasting) between rotating and standing waves. But in the oblique case \textemdash{} Figure \ref{fig:standing_mix_oblique_amp} shows that the amplitude and form of the pulsating wave do not change up until $t = 3000T$ \textemdash{} they seem to be a system's state on their own. Somewhat similar waves were described in \cite{LandsbergKnoblochNew1993}.

\begin{figure}
    \begin{subfigure}[b]{0.48\textwidth}
        \centering
        \includegraphics[width=\linewidth]{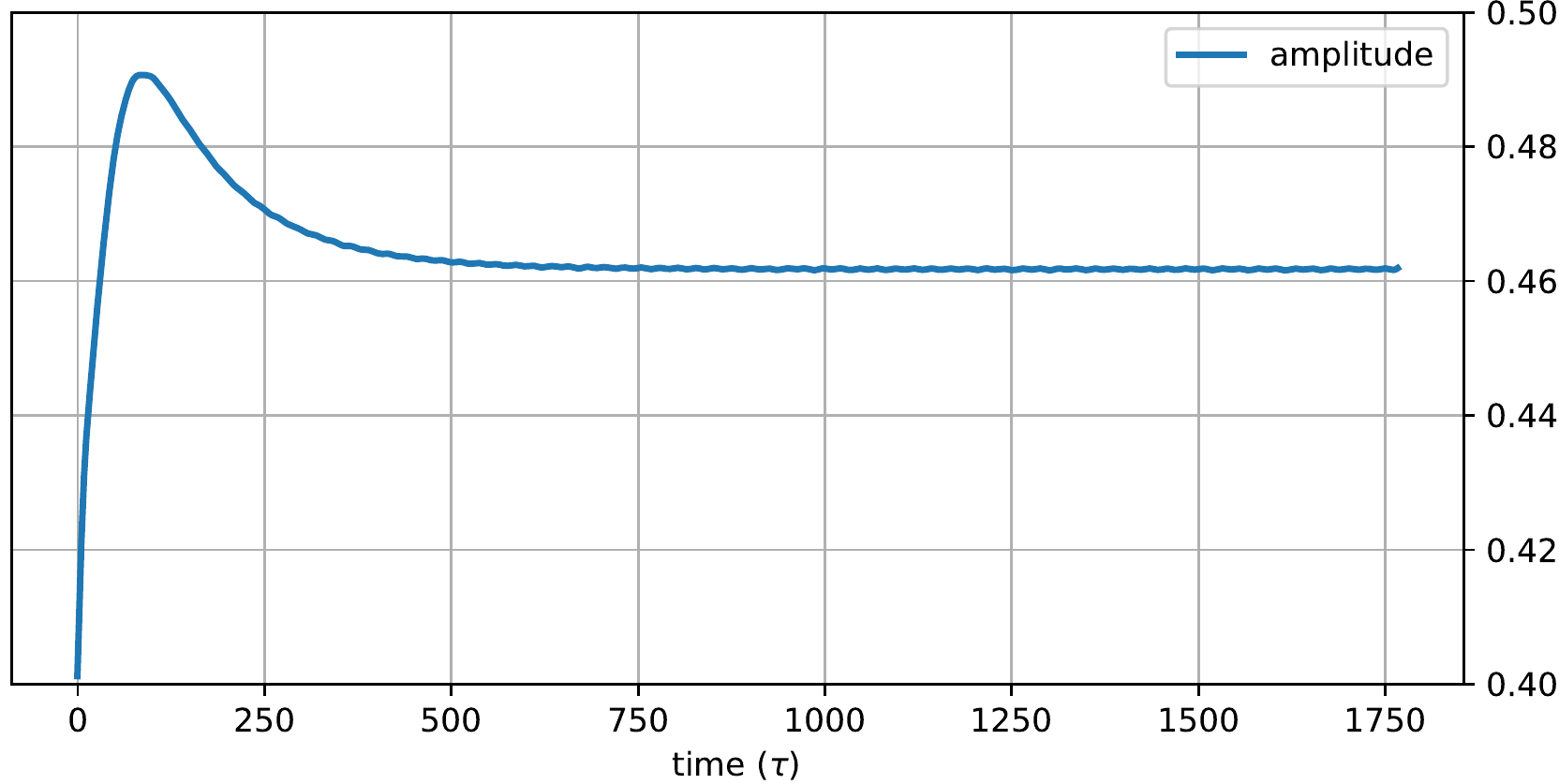}
    \end{subfigure}\hfill%
    \begin{subfigure}[b]{0.48\textwidth}
        \centering
        \includegraphics[width=\linewidth]{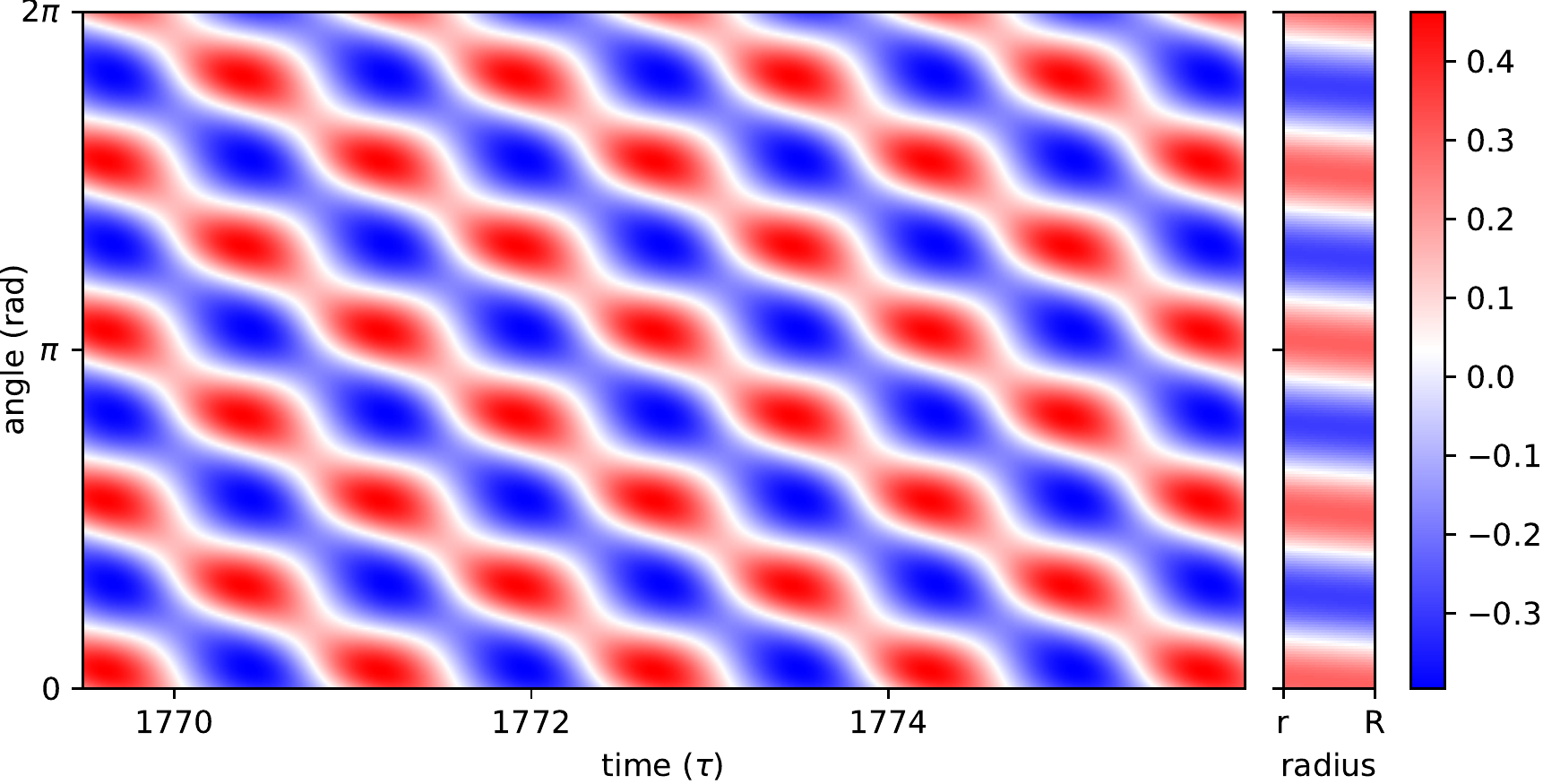}
    \end{subfigure}
    
    \caption{Simulation with standing parameters (Table \ref{tab:params}), oblique boundary conditions $\alpha = \arctan(\nicefrac{3}{4})$, and mixed initial data (Table \ref{tab:ICparams}). Showing the changes of amplitude at $[0, 3000T]$ and a snapshot at $[2989T, 3000T]$.}
    \label{fig:standing_mix_oblique_amp}
\end{figure}

\section{Conclusion}
In the present paper we worked with a delayed scalar diffusion equation of nonlinear optics in a thin annulus with oblique derivative boundary conditions, aiming to predict the existence, shape, and stability properties of spiral waves based on the physical parameters of the model. Our approach consisted in passing to a limiting delayed diffusion equation on a circle, whose rotating and standing waves we can describe by computing\textemdash{}explicitly and in closed form\textemdash{}the coefficients of the Hopf bifurcation normal form. Knowing the relations between `one-dimensional` and `two-dimensional` parameters of the model, we make our predictions about two-dimensional spirals by checking the `one-dimensional` conditions. This allowed us to observe rigidly rotating spirals (corresponding to rotating one-dimensional waves) and pulsating spirals (corresponding to standing one-dimensional waves) in numerical simulations; both types of waves showed some attractivity properties.

This paper can also be seen as a proof of concept that could be applied to general scalar delayed diffusion equations (or systems of reaction-diffusion equations without delay; we just want to have a mechanism leading to a Hopf bifurcation in the limiting problem, be it delay or interactions between several components) in thin domains with oblique derivative boundary conditions. But the corresponding generalisation and rigorous justification are left for the future. 

\paragraph*{Acknowledgements} The authors are grateful to Bernold Fielder for discussions during 11QTDE. The reported study was supported by RFBR according to the research project 18-31-00236

\appendix
\section{Assumptions for the normal form}
\label{appendix:1}
We devote this Appendix to the verification and discussion of the assumptions that are used in \cite{FariaNormal2000} to study the existence of a center manifold and the properties of the normal form on it. 

\subsection{Proof of Theorem \ref{theorem:2}}
Let 
\[
G:\, L^2(\Omega) \longrightarrow L^2(\Omega), \quad \mathcal{D}(G) = \lbrace u \in H^2(\Omega) : \rho u_\rho = \tan\alpha u_\phi\text{ on } \partial\Omega \rbrace, \quad Gu = -\Delta u \,, u \in \mathcal{D}(G),
\]
be the linear operator associated with the oblique derivative boundary value problem in the annulus $\Omega$. 
\begin{lemma}
\label{lemma:A1}
The following is true for the operator $G$:
\begin{enumerate}
    \item has a discrete spectrum $\sigma(G)$;
    \item is closed;
    \item has compact resolvent $R(\lambda, G)$;
    \item all directions in the $\lambda$-plane with the exception of the positive axis are of minimal growth, i.e.
    \[
        \| R(\lambda, G) \|_{\mathcal{L}(L^2(\Omega))} \leq \frac{M_{\arg(\lambda)}}{|\lambda|}, \quad \lambda \notin \mathbb{R_+},
    \]
    for $|\lambda|$ large enough.
\end{enumerate}
\end{lemma}
\begin{proof}
Follows from \cite[Theorem 4.4]{Agmoneigenfunctions1962} as $G$ is associated with an absolutely elliptic regular boundary value problem.
\end{proof}

\begin{lemma}
\label{lemma:A2}
The spectrum $\sigma(G)$ of G lies inside a sector $\lbrace \lambda \in \mathbb{C} : |\arg(\lambda)| \leq |\tan\alpha| \rbrace$.
\end{lemma}
\begin{proof}
The proof basically repeats the reasoning in \cite{KostinSherstyukovcomplex2017}. Consider an eigenvalue problem
\[
    -\Delta u = \lambda u
\]
and a matrix
\[
    Q_\alpha = 
    \begin{bmatrix}
    1 & -\tan\alpha \\
    \tan\alpha & 1
    \end{bmatrix}.
\]
It is easily observed that 
\[
    -\mathrm{div}(\nabla u Q_\alpha^T) = -\Delta u = \lambda u
\]
and that the boundary condition can be written as
\[
    \frac{\partial u}{\partial (Q_\alpha^T\overrightarrow{\nu})} = 0,
\]
where $\overrightarrow{\nu}$ is the outward normal. The divergence theorem then gives 
\[
    \lambda = \int_\Omega \nabla u Q_\alpha^T \nabla u^*
\]
for a normalized eigenfunction $\| u \|_{L^2(\Omega)} = 1$. Then
\[
    |\mathrm{Im}\lambda| = 2 |\tan\alpha| \left| \int_\Omega \mathrm{Im}(u_x u_y^*)  \right| \leq |\tan\alpha| \| \nabla u \|^2_{L^2(\Omega)} = |\tan\alpha| \mathrm{Re}\lambda.
\]
\end{proof}

\begin{lemma}
\label{lemma:A3}
The operator $-G$ is sectorial and generates an analytic semigroup.
\end{lemma}
\begin{proof}
It follows from Lemma \ref{lemma:A1} that $G$ is closed and that resolvent estimates hold for all rays (except for the positive axis), though with different constants. Lemma \ref{lemma:A2} asserts that the spectrum $\sigma(G)$ resides in a sector; hence the resolvent $R(\lambda, G)$ is analytic in the sector's complement and the Phragm\'en–Lindel\"of principle gives an estimate with a uniform constant
\[
    \| R(\lambda, G) \|_{\mathcal{L}(L^2(\Omega))} \leq \frac{M}{|\lambda|}
\]
for all $|\lambda| > 0$ and $|\arg(\lambda)| > |\alpha|$. Thus $-G$ is sectorial and generates an analytic semigroup \cite{EngelNagelshort2006}.
\end{proof}

\begin{proof}[Proof of Theorem \ref{theorem:2}]
Lemma \ref{lemma:A1} shows that $G$ has compact resolvent and Lemma \ref{lemma:A3} says that it generates an analytic semigroup. The proposition follows from Theorem 5.12 of \cite{EngelNagelshort2006}: $G$ generates an immediately compact $C_0$-semigroup.
\end{proof}

\subsection{Linear operator acting on the delayed function}
The reasoning in \cite{FariaNormal2000} is based on the assumption that in the linearized functional differential equation, the delayed part is represented by a bounded linear operator. For us, this is the propagator $\mathcal{B}_{z_0}$ of the oblique derivative boundary value problem for the linear Schr\"odinger equation \eqref{eq:2d_schrodinger},\eqref{eq:bc_oblique}:
\[
    \frac{\partial A}{\partial z} + \mathrm{i} \Delta A = 0, \quad \rho \frac{\partial A}{\partial \rho} = \tan\alpha \frac{\partial A}{\partial \phi} \text{ on } \partial \Omega, \quad \mathcal{B}_{z_0} A_0 = A(z_0; A_0).
\]

Multiplying the equation by $A^*$, integrating over $\Omega$, and taking the real part, we get the "conservation" law,
\[
\frac{1}{2} \frac{d}{dz} \| A \|^2_{L^2(\Omega)} = \tan\alpha \left[ |\gamma_R A|^2_{H^{\nicefrac{1}{2}}(S^1)} - |\gamma_r A|^2_{H^{\nicefrac{1}{2}}(S^1)} \right],
\]
where $\gamma_R$ and $\gamma_r$ are trace operators for the outer and inner circles, and $|\cdot|_{H^{\nicefrac{1}{2}}(S^1)}$ denotes a fractional Sobolev seminorm for functions on a circle:
\[
    v = \sum_{n \in \mathbb{Z}} v_n \frac{e^{\mathrm{i} n \phi}}{\sqrt{2\pi}}, \quad \left| v \right|_{H^{\nicefrac{1}{2}}(S^1)}^2 = \quad  \sum_{n \in \mathbb{Z}} n |v_n|^2.
\]
The equation for the norm of the gradient can be derived similarly
\[
    \frac{1}{2}\frac{d}{dz} \| \nabla A \|^2_{L^2(\Omega)} = \tan\alpha(1 + (\tan\alpha)^2) \left[ |\gamma_R A|^2_{H^{\nicefrac{3}{2}}(S^1)} - |\gamma_r A|^2_{H^{\nicefrac{3}{2}}(S^1)} \right].
\]
If the oblique angle were equal to zero, we would have the standard conservation laws for the linear Schr\"odinger equation,
\[
    \| A(z) \|_{L^2(\Omega)} = \mathrm{const}, \quad \| \nabla A(z) \|_{L^2(\Omega)} = \mathrm{const},
\]
the associated semigroup would be unitary, and the propagator would be bounded. To treat the oblique case, we start with Laplacian eigenfunctions,
\[
    \psi_{n,k} = \frac{y_{n,k}(\rho)}{I_{n,k}} \frac{e^{\mathrm{i}n \phi}}{\sqrt{2\pi}}, \quad \| \psi_{n,k} \|_{L^2(\Omega)} = 1,
\]
where
\[
    y_{n,k}(\rho) = (\zeta_{n,k}r Y_n'(\zeta_{n,k}r) - \mathrm{i}n\tan\alpha Y_n(\zeta_{n,k}))J_n(\zeta_{n,k}\rho) - (\zeta_{n,k}r J_n'(\zeta_{n,k}r) - \mathrm{i}n\tan\alpha J_n(\zeta_{n,k}))Y_n(\zeta_{n,k}\rho)
\]
and $\zeta_{n,k}$ are from Subsection \ref{subsection:laplacian}. The linear Schr\"odinger equation linearly scales eigenfunctions, so
\[
    \frac{1}{2} \frac{d}{dz} \| \mathcal{B}_z \psi_{n,k} \|^2_{L^2(\Omega)} = \tan\alpha \| \mathcal{B}_z \psi_{n,k} \|^2_{L^2(\Omega)} \left[ |\gamma_R \psi_{n,k}|^2_{H^{\nicefrac{1}{2}}(S^1)} - |\gamma_r \psi_{n,k}|^2_{H^{\nicefrac{1}{2}}(S^1)} \right].
\]
With the new notation
\[
    \delta \psi_{n,k} = |\gamma_R \psi_{n,k}|^2_{H^{\nicefrac{1}{2}}(S^1)} - |\gamma_r \psi_{n,k}|^2_{H^{\nicefrac{1}{2}}(S^1)},
\]
we can conclude that
\[
    \| \mathcal{B}_z \psi_{n,k} \|_{L^2(\Omega)} = e^{\tan\alpha \cdot \delta\psi_{n,k}\cdot z}
\]
and so whether the norm grows or diminishes depends on the sign of $\tan\alpha \cdot \delta\psi_{n,k}$.

Let us zoom in on $\delta\psi_{n,k}$. On rewriting it as
\[
  \delta\psi_{n,k} = \frac{n}{I^2_{n,k}}\left[ |y_{n,k}(R)|^2 -|y_{n,k}(r)|^2  \right],
\]
we can guess that just as $\lbrace \zeta_{n,k} \rbrace$ are separated into two qualitatively different groups with $k = 0$ and $k \geq 1$, so could be $\lbrace \delta\psi_{n,k} \rbrace$. Indeed, numerically computed eigenfunctions show that $\delta\psi_{n,k}$ are positive for $k = 0$ and negative for $k \geq 1$ (see Table \ref{tab:psi_nk}).

\begin{table}[h]\centering
\renewcommand{\arraystretch}{1.2}
\begin{tabular}{@{}lllcll@{}}\toprule
    \textbf{Group} & \textbf{Sign} & \textbf{As} $n \to \infty$ & \textbf{As} $k \to \infty$ & \textbf{As} $\varepsilon \to 0$ & \textbf{Wrt} $\alpha$ \\
    
    \midrule
    
    $k = 0$ & $\delta\psi_{n,0} > 0$ & $\delta\psi_{n,0} = \mathcal{O}(n^3)$ & \textemdash{} & $\delta\psi_{n,0} = \mathcal{O}(\varepsilon^2)$ & Even\\
    
    $k \geq1$ & $\delta\psi_{n,k} < 0$ & $\delta\psi_{n,k} = \mathcal{O}(n)$ & $\delta\psi_{n,k} \approx \mathrm{const}$ & $\delta\psi_{n,k} \to -2n$ & Even\\
    
    \bottomrule\\
\end{tabular}
\caption{Properties of $\delta\psi_{n,k}$.}
\label{tab:psi_nk}
\end{table}

As a consequence, depending on the sign of the oblique angle, there are two cases:
\begin{table}[h]\centering
\renewcommand{\arraystretch}{1.2}
\begin{tabular}{@{}lll@{}}\toprule
    \textbf{Angle} & $k = 0$, $z \to +\infty$ & $k \geq 1$, $z \to +\infty$ \\
    
    \midrule
    
    $\tan\alpha > 0$ & $\| \mathcal{B}_z \psi_{n,0} \|_{L^2(\Omega)} \nearrow \infty$ & $\| \mathcal{B}_z \psi_{n,k} \|_{L^2(\Omega)} \searrow 0$ \\
    
    $\tan\alpha < 0$ & $\| \mathcal{B}_z \psi_{n,0} \|_{L^2(\Omega)} \searrow 0$ & $\| \mathcal{B}_z \psi_{n,k} \|_{L^2(\Omega)} \nearrow \infty$ \\
    
    \bottomrule\\
\end{tabular}
\end{table}

In terms of the delayed nonlinear optical system (assume $\tan\alpha > 0$), it shuns memories of radially oscillating states, while those of radially "constant" states are sustained and strengthened. Changing the sign of $\tan\alpha$ inverts the situation.

To extend our understanding from $\lbrace \psi_{n,k} \rbrace$ to other functions, we need some additional properties of them as a system of functions. The general theory says that the generalized eigenfunctions of the oblique derivative Laplacian are complete in $L^2(\Omega)$ \cite{Agmoneigenfunctions1962}. However, we have not proved that there are no other eigenfunctions than $\lbrace \psi_{n,k} \rbrace$, nor have we shown the absence of generalized eigenfunctions. The discussion being informal, we attempt to draw analogies from an oblique derivative boundary value problem in a disc \cite{IlinMoiseevabsence1994,KostinSherstyukovBasis2018} and postulate the following:
\begin{itemize}
    \item $\lbrace \psi_{n,k} \rbrace$ exhaust all the Laplacian eigenfunctions;
    \item there are no generalized eigenfunctions;
    \item for each $n \in \mathbb{Z}$ the subsystem $\lbrace \psi_{n,k} \rbrace_{k \geq 0}$ is a Riesz basis in
    \[
    L^2_n(\Omega)=\lbrace u \in L^2(\Omega): u = e^{\mathrm{i}n \phi} y(\rho),\,y\in L^2(r,R; \rho d\rho) \rbrace;
    \]
    \item $\lbrace \psi_{n,k} \rbrace$ constitute a basis with brackets in $L^2(\Omega)$.
\end{itemize}
The last two points mean that every $u \in L^2(\Omega)$ can be uniquely represented as a sum of mutually orthogonal functions,
\[
    u = \sum_{n \in \mathbb{Z}} u_n, \quad u_n \in L^2_n(\Omega),
\]
and that each $u_n$ in its turn can be uniquely decomposed as
\[
    u_n = \sum_{k \geq 0} u_{n,k} \psi_{n,k} \in L^2_n(\Omega).
\]
Moreover, the following frame conditions hold
\[
    c_n \sum_{k \geq 0}|a_{k}|^2 \leq \| \sum_{k \geq 0} a_{k} \psi_{n,k} \|^2_{L^2(\Omega)} \leq C_n \sum_{k \geq 0}|a_{k}|^2, \quad \lbrace a_k \rbrace \in l^2, \quad 0 < c_n \leq C_n < \infty.
\]
We can thus take a $u \in L^2(\Omega)$ and formally apply the Schr\"odinger propagator $\mathcal{B}_z$ to its decomposition, leading to
\begin{align*}
    \| \mathcal{B}_z u \|^2_{L^2(\Omega)} &= \sum_{n \in \mathbb{Z}} \| \mathcal{B}_z u_n \|^2_{L^2(\Omega)} = \sum_{n \in \mathbb{Z}} \| \sum_{k \geq 0} u_{n,k} \mathcal{B}_z \psi_{n,k} \|^2_{L^2(\Omega)} \\
    & \leq \sum_{n \in \mathbb{Z}} C_n \sum_{k \geq 0} |u_{n,k}|^2 e^{2\tan\alpha \cdot \delta\psi_{n,k}\cdot z} \\
    & \approx \sum_{n \in \mathbb{Z}} C_n \left[ |u_{n,0}|^2 e^{2\tan\alpha M_1^2 n^3 z} + \sum_{k \in \mathbb{N}} |u_{n,k}|^2 e^{-2\tan\alpha M_2^2 n z} \right]
\end{align*}
This suggests that for $\mathcal{B}_z u$ to belong to $L^2(\Omega)$, $u \in L^2(\Omega)$ needs to be very smooth, and that special weighted spaces could be required to make analysis rigorous.

\section{Numerical methods}
\label{appendix:2}
In this Appendix we shall describe the numerical methods that we used to compute the data presented in the figures. 

\subsection{Helmholtz equation}
The first building block is the numerical solution of
\[
    (-\Delta + c)u = f
\]
in an annulus subject to oblique derivative boundary conditions. If the boundary conditions were Neumann, we could easily solve this using the Fourier method, which relies on the fact that Neumann-Laplacian eigenfunctions form an orthogonal basis in $L^2$. For oblique boundary conditions, the situation is different: it is known that the bi-orthogonal system of eigenfunctions is complete for Laplacian in a disc but is not a basis \cite{IlinMoiseevabsence1994}. However, a more refined property holds as this system forms a basis with brackets \cite{KostinSherstyukovBasis2018}, which is enough to justify the Fourier method.

We can then take a uniform polar grid with $N_\rho$ and $N_\theta$ radial and angular knots, apply Fast Fourier Transform to $f$ on each angular slice, and solve a collection of tridiagonal systems of linear equations
\[
    (A_k + cI)u_k = f_k, \qquad k = -\frac{N_\theta}{2},\ldots,\frac{N_\theta}{2}-1,
\]
followed by the inverse Fast Fourier Transform. Matrices $A_k = A + P_k \in \mathbb{C}^{N_\rho \times N_\rho}$ are discretizations of the Laplacian on the corresponding harmonic, where
\[
A = \begin{bmatrix}
    \frac{2}{\delta\rho^2} & -\frac{2}{\delta\rho^2} & 0 & \dots & 0 & 0 & 0 \\
    -\frac{1}{\delta\rho^2} + \frac{1}{2\delta\rho(r + \delta\rho)} & \frac{2}{\delta\rho^2} & -\frac{1}{\delta\rho^2} - \frac{1}{2\delta\rho(r + \delta\rho)} & \dots & 0 & 0 & 0 \\
    \hdotsfor{7} \\
    0 & 0 & 0 & \dots & -\frac{1}{\delta\rho^2} + \frac{1}{2\delta\rho(R - \delta\rho)} & \frac{2}{\delta\rho^2} & -\frac{1}{\delta\rho^2} - \frac{1}{2\delta\rho(R - \delta\rho)}\\
    0 & 0 & 0 & \dots & 0 & -\frac{2}{\delta\rho^2} & \frac{2}{\delta\rho^2} \\
\end{bmatrix}
\]
and
\[
P_k = \mathrm{diag}\left(
\frac{k}{r^2} \left[k + \mathrm{i}\tan{\alpha}\frac{2r-\delta\rho}{\delta\rho}\right], \frac{k^2}{(r + \delta\rho)^2},
\ldots,
\frac{k^2}{(R - \delta\rho)^2},
\frac{k}{R^2} \left[k - \mathrm{i}\tan{\alpha}\frac{2R+\delta\rho}{\delta\rho}\right]
\right).
\]

\subsection{Linear Schr\"odinger equation}
To evaluate the nonlinear term of the equation, we need to solve an initial-boundary value problem for 
\[
    \frac{\partial A}{\partial z} + \mathrm{i} \Delta A = 0, \qquad A(z = 0) = A_0
\]
from $z = 0$ up until $z = z_0$. We use the following higher-order scheme to propagate along $\delta z$
\[
    A_{+} = \left(I + \mathrm{i} \frac{\delta z}{2} \Delta\right)^{-1} \left(I - \mathrm{i} \frac{\delta z}{2} \Delta\right) A_{-}
\]
since it is crucial to have an accurate value of the nonlinearity.

\subsection{Main diffusion equation}
To make long-time simulations we use the standard implicit Euler scheme. Because of the delay, the nonlinearity need not be approximated with iterations. To produce figures for this paper, we used $N_\theta = 256$ and $N_\rho = 128$; per each delay interval $[nT, (n+1)T]$, we made $N_T = 180$ steps for rotating parameters and $N_T = 60$ for standing parameters (see Table \ref{tab:params}).

\bibliography{zotero} 
\bibliographystyle{plain}
\end{document}